\documentclass[12pt,draftclsnofoot,onecolumn]{IEEEtran}

\newcommand\R{\mathrm{R}}

\def\tran{^{\mbox{\scriptsize T}}}

\newcommand\zbar{\underline{z}}
\newcommand\bbar{\underline{\tilde{b}}}

\usepackage{dsfont}
\usepackage{moreverb}
\usepackage{epsfig}
\usepackage{amsmath,amssymb,amsthm,mathrsfs,amsfonts,dsfont}
\usepackage{adjustbox,lipsum}
\usepackage{algorithm,algorithmic}
\usepackage{amsfonts}
\usepackage{epsfig}
\usepackage{amssymb}
\usepackage{amsmath}
\usepackage{amsthm}
\usepackage{multirow}
\usepackage{rotating}
\usepackage{graphicx}
\usepackage{tabularx}
\usepackage{array}
\usepackage{anyfontsize}
\usepackage{color,soul}
\usepackage{graphicx,dblfloatfix}
\usepackage{epstopdf}
\usepackage{blindtext}
\usepackage{amsmath}
\usepackage{amsthm,amssymb,amsmath,bm}
\usepackage{amsfonts}
\usepackage{epsfig}
\usepackage{amssymb}
\usepackage{amsmath}
\usepackage{cite}
\usepackage{graphicx}
\usepackage{fancyhdr}
\usepackage{subfigure}
\usepackage[font={small}]{caption}
\usepackage{tabularx}
\usepackage{tcolorbox}
\usepackage{cite}
\usepackage{algorithm}
\usepackage{algorithmic}
\usepackage{soul}
\usepackage{setspace}
\usepackage{enumitem}

\usepackage{arydshln}

\DeclareMathOperator*{\argmax}{arg\,max}
\DeclareMathOperator*{\argmin}{arg\,min}
\newtheorem{theorem}{Theorem}
\newtheorem{lemm}{Lemma}
\newtheorem{Definition}{Definition}

\newtheorem{Pro}{Proposition}

\def\blue{\textcolor{blue}}

\def\brown{\textcolor{brown}}

\def\gray{\textcolor{gray}}

\allowdisplaybreaks

\graphicspath{{Figure/}}

\usepackage[hidelinks]{hyperref}

\begin{document}
	\sloppy

	\title{Status Updating under Partial Battery Knowledge in Energy Harvesting IoT Networks}
	
	\author{\IEEEauthorblockN{Mohammad Hatami\IEEEauthorrefmark{1}, Markus Leinonen\IEEEauthorrefmark{1}, and Marian Codreanu\IEEEauthorrefmark{2}}
		\thanks{
			\IEEEauthorrefmark{1}
			Centre for Wireless Communications -- Radio Technologies, University of Oulu, Finland. e-mail: firstname.lastname@oulu.fi
			
			\IEEEauthorrefmark{2}
			Department of Science and Technology, Link\"{o}ping University, Sweden. e-mail: marian.codreanu@liu.se.

			This research has been financially supported by the Infotech Oulu, the Academy of Finland (grant 323698), and Research Council of Finland (former Academy of Finland) 6G Flagship Programme (Grant Number: 346208). The work of M. Leinonen has also been financially supported in part by the Academy of Finland (grant 340171).
			
		}
	}
	

	\maketitle

	\begin{abstract}
		We study status updating under inexact knowledge about the battery levels of the energy harvesting  sensors in an IoT network, where users make 
		on-demand 
		requests to a cache-enabled edge node to send updates about various random processes monitored by the sensors. 
		To serve the request(s), the edge node either commands the corresponding sensor to send an update or uses the aged data from the cache. 
		We find a control policy that minimizes the average on-demand AoI subject to per-slot energy harvesting constraints under partial battery knowledge at the edge node.
		Namely, the edge node is informed about sensors' battery levels only via received status updates, leading to uncertainty about the battery levels for the decision-making. 
		We model the problem as a POMDP which is then reformulated as an equivalent belief-MDP. The belief-MDP in its original form is difficult to solve due to the infinite belief space.
		However, by exploiting a specific pattern in the evolution of beliefs, we truncate the belief space and develop a  dynamic programming algorithm to obtain an optimal policy.
		Moreover, we address a multi-sensor setup under a transmission limitation for which we develop an asymptotically optimal algorithm.
		Simulation results assess the performance of the proposed methods.
	\end{abstract}
	\section{Introduction}

	In future Internet of things (IoT) systems,
	timely delivery of status updates about a remotely monitored random process to a destination is the key enabler for the emerging time-critical applications, e.g., {drone control, smart home, and transport systems.} 
	Such destination-centric information freshness can be quantified by the \textit{age of information} (AoI) \cite{AoI_Orginal_12,sun2019age}. 
	On the other hand,  IoT networks with low-power sensors are subject to stringent energy limitations, which is often counteracted by \textit{energy harvesting} (EH) technology. 
	To summarize, these emerging applications require designing \textit{AoI-aware status updating control} that both guarantees timely status delivery and accounts for the limited energy resources of EH sensors.

	In this paper, we consider a status update IoT network consisting of EH sensors, users, and an edge node, which acts as a gateway between the sensors and users, as depicted in Fig.~\ref{fig_systemmodel_multisensor}. The users are interested in time-sensitive information about several random processes, each measured by a sensor.
	The users send requests to the edge node that has a cache storage to store the most recently received status update from each sensor. To serve a user's request, the edge node either commands the corresponding sensor to send a fresh status update or uses the aged data from the cache. 
	This introduces an inherent trade-off between the 
	age of information
	(AoI) at the users and the energy consumption of the sensors.
	As the main novelty of our work compared to the related AoI-aware network designs \cite{Hatami-etal-20,hatami2020aoi,hatami2021spawc,hatami2022JointTcom}, we consider a practical scenario where the edge node is informed of the sensors' battery levels only via the received status updates, 
	leading to \textit{partial} battery knowledge at the edge node. Particularly, our objective is to find the best actions of the edge node to minimize the average AoI of the served measurements, i.e., \textit{average on-demand AoI}. Accounting for the {partial} battery knowledge, we model this as an average-cost partially observable Markov decision process (POMDP). We then convert the POMDP into a belief-state MDP and, {via characterizing its key structures}, develop {an iterative algorithm to obtain} an optimal policy. 
	Further, we extend the proposed approach to the multi-sensor setup under a transmission constraint, where only a limited number of sensors can send status updates at each time slot.
	{Numerical experiments assess the performance of the proposed methods.}
	
	\subsection{Contributions}
	{The primary contributions of our study are summarized as follows:}
	\begin{itemize}
		\item We consider (on-demand) AoI-minimization for a status update IoT network where the decision-maker does not know the exact battery levels of the sensors at each  slot.  Accounting for the partial battery knowledge, we model the problem as an average-cost POMDP.
		\item We reformulate the POMDP into an equivalent belief-MDP which, however, is difficult to solve in its original form due to the infinite belief space. Fortunately, we exploit  a certain pattern in the evolution of beliefs to truncate the belief space and develop a dynamic programming algorithm  that obtains an optimal policy.
		{In addition, we derive an efficient algorithm implementation by exploiting the inherent sparsity of the transition matrices.}
		\item Further, we extend the proposed approach to the multi-sensor setup under a transmission constraint, where only a limited number of sensors can send status updates at each time slot. {In particular, we develop 
			a low-complexity relax-then-truncate algorithm and show 
			its asymptotic optimality  as the number of sensors approaches infinity.}
		\item
		Numerical experiments illustrate the threshold-based structure of an optimal policy and show the gains obtained by the proposed optimal POMDP-based policy compared to a request-aware greedy policy.
		Further, numerical experiments depict that the proposed relax-then-truncate method has near-optimal performance even for moderate numbers of sensors in multi-sensor scenarios under a transmission constraint.
	\end{itemize}

	{To the best of our knowledge, this is the first work that derives an optimal policy for
		AoI minimization in a network with EH sensors, where the decision-making relies only on partial battery knowledge about the sensors' battery levels. }
	\subsection{Related Works}

	AoI-aware scheduling has witnessed a great research interest the last few years.
	The works 
	\cite{Ceran2019HARQ,Ceran2021MultiUserCMDP,tang2020CMDPTRUNCATE,Kadota_Modiano2018Broadcast,Maatouk2021OptimalityWhittle,kriouile2021global,Hsu_modiano2020aoimultiuser_tcm,zakeri2023TWC,yao2022ageTMCJournal,gong2020AoI-Random-Arrival,shao2020partially,stamatakis2022semantics,liu2022PartialObservations,SaeidSadeghi2023MultiSourceAoI} consider a sufficient power source 
	whereby an update can be sent any time.
	Differently,
	\cite{wu2017optimal_oneunitenergy,Stamatakis2019control,ceran2021learningEH,tunc2019optimal,leng2019AoIcognitive,Elvina2021SourceDiversityEH,abd2019reinforcement,AbdElmagid2020AoIOptimalJoint,Hatami-etal-20,leng2022LearningTransmitEH,ArafaTimelyErasureEHMultipleSource}
	consider that the {source nodes} are powered by {energy harvested from the environment}; thus, AoI-aware scheduling is carried out under the energy causality constraint at the source nodes.
	Also, while the above works (implicitly) assume that time-sensitive information
	is needed at the destination at all time moments, \cite{hatami2020aoi,hatami2021spawc,Hatami2022AsymptoticallyOpt,hatami2022JointTcom,hatami2022spawc_partialbattery,chiariotti2021query,Li2021waiting} 
	study information freshness of the source(s) driven by {users' requests}. 
	{Particularly, in our prior research \cite{hatami2020aoi,Hatami-etal-20,hatami2021spawc,hatami2022JointTcom}, we introduced the concept of on-demand AoI. This metric quantifies the freshness of information seen by users in request-based status updating systems.
		In \cite{hatami2020aoi,Hatami-etal-20,hatami2021spawc,hatami2022JointTcom}, we have mainly focused on optimal scheduling under the assumption that the decision-maker (i.e., the edge node) possessed precise knowledge of the sensors' battery levels at every time slot. However, such an assumption necessitates continuous coordination between the sensors and the edge node, which may not always be feasible in practical scenarios.
		In contrast, this study delves into optimal scheduling under \textit{partial battery knowledge at the edge node}, a scenario that can be effectively modeled as a POMDP.}
	


	
	A few works have applied POMDP formulation, in which the state of the system is not fully observable to the decision maker, in AoI-aware design
	\cite{yao2022ageTMCJournal,gong2020AoI-Random-Arrival,shao2020partially,stamatakis2022semantics,leng2019AoIcognitive}. In
	\cite{yao2022ageTMCJournal}, the authors proposed POMDP-based AoI-optimal transmission scheduling in a status update system under an average energy constraint and uncertain channel state information.
	In \cite{gong2020AoI-Random-Arrival}, the authors proposed 
	an age-aware scheduling policy for a multi-user uplink system under partial knowledge of the status update arrivals at the monitor node.
	{In \cite{shao2020partially}, the authors investigated AoI-optimal scheduling in a wireless sensor network where the AoI values of the sensors are not directly observable by the access point.}
	In \cite{stamatakis2022semantics}, the authors derived an optimal sensor probing policy in an IoT network with intermittent faults and inexact knowledge about the status (healthy or faulty) of the system. 
	In \cite{liu2022PartialObservations}, the authors derived age-aware POMDP-based
	scheduling for a wireless multi-user uplink network with partial observations of the local ages at end devices. 
	In \cite{leng2019AoIcognitive}, the authors investigated AoI minimization for an EH cognitive secondary user with either perfect or imperfect spectrum sensing.
	Preliminary results of this paper appear in \cite{hatami2022spawc_partialbattery}.
	\subsection{Organization}
	The structure of the paper is outlined as follows. For the sake of clarity in presentation, we first restrict ourselves to the single-sensor scenario\footnote{This is equivalent to the case where multiple sensors have independent links to the edge node.
	}, and then, we address the multi-sensor scenario under the transmission constraint. In particular, Section~\ref{sec_systemmodel_single_sensor} describes the single-sensor system setup and the problem formulation.
	In Section~\ref{sec_single_sensor}, we propose a novel POMDP-based approach that finds an optimal policy for the single-sensor setup. In Section~\ref{sec_multisensor-limited_bw}, we address the multi-sensor scenario under a transmission constraint. In Section~\ref{sec_simulation}, we evaluate the performance of the proposed methods through simulations. Finally, in Section~\ref{sec_conclusions}, we conclude the paper.

	\section{Single-Sensor System Model and Problem Formulation}\label{sec_systemmodel_single_sensor}

	\begin{figure}[t!]
		\centering
		\includegraphics[width=.8\columnwidth]{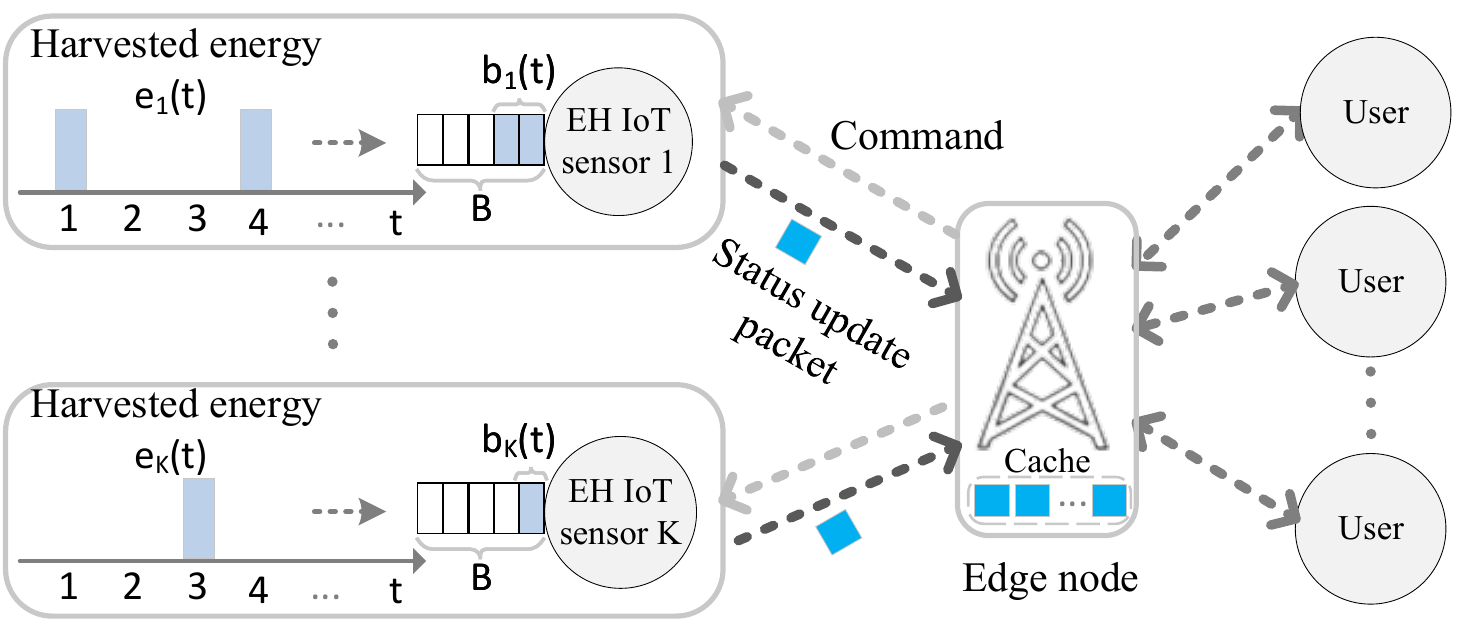}
		\caption{An IoT sensing network with $K$ EH sensors, an edge node, and users, which are interested in timely status update information of the physical processes monitored by the sensors.
		}
		\label{fig_systemmodel_multisensor}
	\end{figure}

	\subsection{Network Model}\label{sec_network}
	We consider a status update system, where an energy harvesting (EH) sensor {(e.g., sensor~1 in Fig.~\ref{fig_systemmodel_multisensor})} sends status updates about the monitored random process to users via a cache-enabled edge node, which acts as a gateway between the sensors and the users. 
	A time-slotted system with slots ${t \in \mathbb{N}}$ is considered. 
	We consider request-based status updating, where, at the beginning of slot $t$, users request for the status of the sensor (i.e., a new measurement) from the edge node. The edge node, which has a cache that stores the most recently received \textit{status update} from the sensor,  handles the arriving requests during the same slot $t$.
	Let ${r(t) \in \{0,1\}}$, ${t=1,2,\dots}$, denote the random process of requesting the status of the sensor at slot $t$; ${r(t) = 1}$ if the status is requested (by at least one user) and ${r(t)=0}$ otherwise. The requests are independent across time slots and the probability of having a request at each time slot is ${p=\mathrm{Pr}\{r(t) = 1\}}$. Upon receiving a request at slot $t$, the edge node serves the requesting user(s) by either 1) commanding the sensor to send a fresh status update packet\footnote{{In this paper, the terms ``\textit{status update packet}'', ``\textit{status update}'', and ``\textit{update}'' are used interchangeably.}} or 2) using the stored measurement from the cache. Let ${a(t) \in \mathcal{A} = \{0,1\}}$ be the \textit{command action of the edge node} at slot $t$; ${a(t)=1}$ if the edge node commands the sensor to send an update and ${a(t)=0}$ otherwise.
	
	
	\subsection{Energy Harvesting Sensor}\label{EH_model}
	{The sensor operates by harvesting energy from the environment and storing it into a battery of finite capacity $B$ (units of energy).}
	We model the energy arrivals ${e(t) \in \left\lbrace 0 ,1\right\rbrace }$, ${t = 1,2,\dots}$, as a Bernoulli process with rate ${\lambda=\Pr\{ e(t) = 1 \}}$, $\forall t$.
	{This characterizes the discrete nature of the energy arrivals in a slotted-time system, i.e., at each time slot, a sensor either harvests one unit of energy or not (see e.g., \cite{hatami2020aoi,hatami2022JointTcom,Stamatakis2019control,bernoulliEH_Baknina2018_sending})}.
	{We denote the battery level of the sensor at the beginning of slot $t$ by ${b(t) \in \{0,\ldots,B\}}$.
		We assume that measuring and transmitting a status update from the sensor to the edge node consumes one unit of energy (see, e.g., \cite{hatami2020aoi,hatami2022JointTcom,yao2022ageTMCJournal,Stamatakis2019control,wu2017optimal_oneunitenergy}).}
	Thus, if the sensor is commanded to send an update (i.e., $a(t)=1$), it can only do so if its battery is not empty (i.e., $b(t)\geq 1$).
	Let $d(t)\in\{0,1\}$ indicate the sensor's action at slot $t$; $d(t)=1$ if a status update is sent, and $d(t)=0$ otherwise.
	Thus, ${d(t) =  a(t) \mathds{1}_{\{b(t) \geq 1\}}}$, where ${\mathds{1}_{\{\cdot\}}}$ is the indicator function. 
	Finally, the evolution of the battery level is given by
	\begin{equation}\label{battery_evo}
		b(t+1) = \min\left\lbrace  b(t)+e(t)-d(t) , B \right\rbrace.
	\end{equation}
	\subsection{Status Updating with Partial Battery Knowledge}\label{sec_system_model_partial_battery}
	{As the main distinctive feature of this paper, we consider}
	a practical operation mode of the network in which the edge node is informed about the sensor's battery level (only) via the received \textit{status update packets}. Specifically, each status update packet contains the measured value (status) of the physical quantity, a time stamp representing the time when the sample was generated, and the current battery level of the sensor. 
	{As the inevitable consequence of this status updating procedure}, the edge node has only \textit{partial} knowledge about the battery level at each slot, i.e.,  \textit{outdated} knowledge based on the sensor’s last update. It is worth emphasizing that considering this realistic setting is in stark contrast to the previous works on AoI-aware network designs (see e.g., \cite{hatami2021spawc,Hatami2022AsymptoticallyOpt,hatami2022JointTcom,abd2019reinforcement}) which all assume that perfect battery knowledge is available at the decision-maker (herein, the edge node) at each slot.

	Formally, let ${\tilde{b}(t) \in \{1,2,\dots,B\}}$ denote the edge node's \emph{knowledge} about the sensor's battery level at slot $t$. 
	At slot $t$, let $u(t)$ denote the most recent slot in which the edge node received a status update packet, i.e., ${u(t) = \max \{t'| t'<t, d(t') = 1 \}}$.
	Thus, the true battery level and the knowledge about the level are interrelated as ${\tilde{b}(t) = b(u(t))}$. Specifically, at slot $t$, $\tilde{b}(t)$ indicates what the sensor's battery level was at the beginning of the most recent slot at which the edge node received a status update. Henceforth, we refer to $\tilde{b}(t)$ 
	as 
	the \textit{partial battery knowledge}.
	

	\subsection{On-Demand Age of Information}\label{sec_AoI}
	{We use the \textit{on-demand AoI} metric \cite{hatami2020aoi,hatami2021spawc} to measure the freshness of information seen by the users in our request-based status updating system.}
	Let $\Delta(t)$ be the AoI about the monitored  process at the edge node at the beginning of slot $t$, i.e., the number of slots elapsed since the generation of the latest received update, which 
	is expressed as
	${\Delta(t) = t - u(t)}$.
	We make a common assumption (see e.g., 
	\cite{shao2020partially,chiariotti2021query,AbdElmagid2020AoIOptimalJoint,abd2019reinforcement,Elvina2021SourceDiversityEH,leng2019AoIcognitive,tunc2019optimal,ceran2021learningEH,Maatouk2021OptimalityWhittle,Ceran2019HARQ})
	that $\Delta(t)$ is upper-bounded by a sufficiently large value $\Delta^{\mathrm{max}}$, i.e., ${\Delta(t) \in \{1, 2,\ldots ,\Delta^{\mathrm{max}}\}}$.
	{In addition to tractability, this makes further counting unnecessary once the available measurement becomes excessively outdated.}
	The evolution of $\Delta(t)$ is given by
	\begin{equation}\label{eq_AoI}
		\Delta(t+1)= \left\{
		\begin{array}{ll}
			1,&\mathrm{if} ~d(t)=1, \\
			\min \{\Delta(t)+1,\Delta^{\mathrm{max}}\},&\mathrm{if}~d(t)=0,
		\end{array}
		\right.
	\end{equation}
	which can be written in a compact form as $\Delta(t+1)=\min \{ (1-d(t)) \Delta(t)+1,\Delta^{\mathrm{max}}\}$. 
	
	We define on-demand AoI at slot $t$ as 
	\begin{equation}\label{on-demand-AoI}
		\Delta^\mathrm{OD}(t) 
		\triangleq r(t) \Delta(t+1)
		= r(t) \min \{ (1-d(t)) \Delta(t)+1,\Delta^{\mathrm{max}}\}.
	\end{equation}
	Referring to \eqref{on-demand-AoI}, the requests are made at the beginning of slot $t$ and measurements are sent by the edge node at the end of the same slot, thus, $\Delta(t+1)$ is the AoI perceived by the users.

	\subsection{Problem Formulation}\label{sec_cost}
	
	We aim to find  the best action of the edge node at each time slot, i.e., $a(t)$, $t = 1,2,\ldots$, called an \textit{optimal policy}, that minimizes the average cost (i.e., average on-demand AoI), defined as
	\begin{equation}\label{eq_average_cost_persensor}
		\bar{C}=\lim_{T\rightarrow\infty} \frac{1}{T}\sum_{t=1}^{T} \mathbb{E} [\Delta^\mathrm{OD}(t)],
	\end{equation}
	where the expectation is taken over {all system dynamics, i.e., random process of energy arrivals and requests, as well as the (possibly randomized) policy constructed in response to the requests.}

	\section{POMDP Modeling, Optimal Policy, and Proposed Algorithm}\label{sec_single_sensor}
	We model the problem of finding an optimal policy as a partially observable Markov decision process (POMDP) and develop an iterative algorithm to find such an optimal policy.

	\subsection{POMDP Modeling}\label{sec-OMDP-Modeling} 

	The POMDP is defined by a tuple $(\mathcal{S},\mathcal{O},\mathcal{A},\Pr(s(t+1)| s(t), a(t)),\Pr(o(t)| s(t), a(t-1)),c(s(t),a(t)))$ \cite[Chap.~7]{sigaud2013markov}, with the following elements.
	\begin{itemize}
		\item \textit{State Space $\mathcal{S}$}: Let ${s(t) \in \mathcal{S}}$ denote the system state at slot $t$, which we define as ${s(t) = (b(t),r(t), \Delta(t),\tilde{b}(t))}${, where ${b(t) \in \{0, 1,\ldots, B\}}$ is the battery level, ${r(t) \in \{0,1\}}$ is the request indicator, ${\Delta(t) \in \{1,2,\ldots,\Delta^{\mathrm{max}}\}}$ is the AoI, and ${\tilde{b}(t)\in \{1,2,\dots,B\}}$ is the partial battery knowledge.} 
		The state space $\mathcal{S}$ has a finite dimension ${|\mathcal{S}| = 2B(B+1)\Delta^{\mathrm{max}}}$. We denote the \textit{observable} part of the state (i.e., visible by the edge node) by ${s^{\mathrm{v}}(t) = (r(t), \Delta(t),\tilde{b}(t))}$; thus, ${s(t) = (b(t),s^{\mathrm{v}}(t))}$.
		\item \textit{Observation Space $\mathcal{O}$}:
		Let ${o(t) \in \mathcal{O}}$ be the edge node's \textit{observation} about the system state at slot $t$. We define it as the visible part of the state, i.e., ${o(t) = s^\mathrm{v}(t)}$. The observation space $\mathcal{O}$ has a finite dimension $ |\mathcal{O}| = 2B\Delta^\mathrm{max}$. 
		\item \textit{Action Space $\mathcal{A}$}: At each slot, the edge node decides whether to command the sensor
		or not, i.e., ${a(t) \in \mathcal{A} = \{0,1\}}$.
		\item \textit{State Transition Probability $\Pr(s(t+1)| s(t), a(t))$}: {The state transition probability specifies the probability of transitioning from current state ($s(t)$) ${s =(b, r, \Delta, \tilde{b})}$ to next state ($s(t+1)$) ${s^\prime = (b^\prime ,r^\prime,  \Delta^\prime,\tilde{b}^\prime)}$ when taking a particular action ${a(t) = a}$, which is given by}
		\begin{equation}\label{eq_stp}
			\begin{array}{ll}
				\Pr(b^\prime,r^\prime,  \Delta^\prime,\tilde{b}^\prime \mid b, r,  \Delta,\tilde{b} , a ) = 
				\Pr \big(r^\prime\big) \Pr(b^\prime,\mid b , a) \Pr(\Delta^\prime,\tilde{b}^\prime \mid b,\tilde{b}, \Delta,a),
			\end{array}
		\end{equation}
		where
		\begin{equation}
			{\Pr(r^\prime) = pr^\prime+(1-p)(1-r^\prime),~~~r^\prime\in \{0,1\}}
		\end{equation}
		\begin{equation}
			\begin{aligned}
				&\Pr (b^\prime\mid b = B , a = 0) = \mathds{1}_{\{b^\prime=B\}},~b^\prime \in \{0,1,\dots,B\},\\
				&\Pr (b^\prime\mid  b < B , a = 0) = \left\lbrace 
				\begin{array}{ll}
					{\lambda,} & b^\prime = b + 1,\\
					{1-\lambda,} & b^\prime = b ,\\
					0, & \mbox{otherwise.}
				\end{array}
				\right.\\
				&\Pr (b^\prime \mid b = 0 , a = 1) = \left\lbrace 
				\begin{array}{ll}
					{\lambda,} & b^\prime = 1,\\
					{1-\lambda,} & b^\prime = 0, \\
					0, & \mbox{otherwise.}
				\end{array}
				\right.\\
				&\Pr (b^\prime\mid  b\geq1 , a = 1) = \left\lbrace
				\begin{array}{ll}
					{\lambda,} & b^\prime = b,\\
					{1-\lambda,} & b^\prime = b - 1,\\
					0, & \mbox{otherwise.}
				\end{array}
				\right.
			\end{aligned}
		\end{equation}
		\begin{equation}
			\begin{aligned}
				&\Pr (\Delta^\prime, \tilde{b}^\prime\mid  b,\tilde{b}, \Delta, a = 0) = \mathds{1}_{\{\Delta^\prime= \min\{\Delta+1,\Delta^{\mathrm{max}}\}, \tilde{b}^\prime = \tilde{b}\}},\\
				&\Pr (\Delta^\prime, \tilde{b}^\prime\mid  b = 0 , \tilde{b}, \Delta, a = 1 ) = \mathds{1}_{\{\Delta^\prime= \min\{\Delta+1,\Delta^{\mathrm{max}}\}, \tilde{b}^\prime = \tilde{b}\}},\\
				&\Pr (\Delta^\prime, \tilde{b}^\prime\mid  b \geq 1 , \tilde{b}, \Delta, a = 1 ) = \mathds{1}_{\{\Delta^\prime= 1, \tilde{b}^\prime = b\}}.
			\end{aligned}
		\end{equation}
		\item \textit{Observation Function ${\Pr(o(t) \mid s(t), a(t-1))}$}: The observation function is the
		probability of observing $o(t)$ given a state $s(t)$
		and an action $a(t-1)$. In our model, this is given by ${\Pr(o(t)| s(t), a(t-1))=\Pr ({o(t)}| b(t),s^{\mathrm{v}}(t),a(t-1)) = \mathds{1}_{\{ o(t) = s^{\mathrm{v}}(t)\}}}$.
		\item \textit{Immediate Cost Function $c(s(t),a(t))$}:
		The immediate cost of taking action $a(t)$ in state ${s(t)=(b(t),r(t),\Delta(t),\tilde{b}(t))}$ is ${c(s(t),a(t)) = r(t)\min\{(1-a(t)\mathds{1}_{\{b(t)\geq1\}})\Delta(t)+1,\Delta^{\mathrm{max}}\}}$.
	\end{itemize}
	
	\subsection{Belief-State}\label{sec-belief-def} 
	As the POMDP formulation above implies, the system state $s(t)$ is not fully observable by the edge node -- the decision-maker -- at slot $t$. To reiterate, the state consists of two parts as ${s(t)=(b(t),s^{\mathrm{v}}(t))}$. Consequently, at slot $t$, the exact battery level $b(t)$ is unknown to the edge node, whereas the requests, AoI, and partial battery knowledge -- captured by $s^{\mathrm{v}}(t)$ -- are observable. This incomplete state information in a POMDP causes challenges for (optimal) status updating because the edge node can make decisions only based on the available
	information or on the quantities derived from that information. To counteract such insufficiency in the state information, we need to define state-like quantities that preserve the Markov property and summarize all the necessary information for the edge node pertaining to finding an optimal policy. These are called \textit{sufficient information states} \cite[Chapter~7]{sigaud2013markov}, 
	and defined as follows. 
	

	\begin{Definition}[Sufficient information state {\cite[Chapter~7]{sigaud2013markov}}]
		Let $\phi^\mathrm{c}(t)$ be the \textit{complete information state} at slot $t$, which consists of an initial probability distribution over the states and the history of observations and actions starting from ${t = 1}$, i.e., ${\{o(1),\dots,o(t),a(1),\dots,a(t-1)\}}$.
		Let $\phi(t)$ be any information state derived from $\phi^\mathrm{c}(t)$. The sequence $\{\phi(t)\}$ is said to be \textit{sufficient} in regard to finding the optimal policy when, for any slot $t$, it satisfies 
		\begin{equation}\notag
			\begin{aligned}
				&1)~\phi(t) = f(\phi(t-1),o(t),a(t-1)),\\
				&2)~\Pr(s(t)\mid \phi(t)) = \Pr(s(t) \mid \phi^\mathrm{c}(t)), \\
				&3)~\Pr(o(t) \mid \phi(t-1),a(t-1)) =  \Pr(o(t) \mid \phi^\mathrm{c}(t-1),a(t-1)),
			\end{aligned}
		\end{equation}
		where $f(\cdot)$ is an update function defining the information state process $\{\phi(t)\}$.
	\end{Definition}
	
	
	One sufficient information state is a \textit{belief-state}. We define the belief-state at slot $t$ as ${z(t) = \left({\beta}(t), s^{\mathrm{v}}(t)\right) \in \mathcal{Z}}$, where $\beta(t)$ is \textit{belief} about the battery level\footnote{In general, the belief associated with a POMDP is a probability distribution over the \textit{entire} state space $\mathcal{S}$. However, because $s^{\mathrm{v}}(t)$ is fully observable in our problem, it has no uncertainty to be modelled via a belief. {In fact, this particular type of a POMDP encountered in this paper is (sometimes) called a mixed observable MDP (MOMDP) \cite{araya2010closerMOMDP}.}} $b(t)$ and $\mathcal{Z}$ is the belief-state space. The belief is a $({B+1})$-dimensional vector ${\beta(t) = (\beta_{0}(t),\dots,\beta_{B}(t))\tran{\in \mathcal{B}}}$, $\sum_{j = 0}^{B} \beta_j(t) = 1$, that gives the probability distribution on the possible values of the battery levels at slot $t$, where ${\mathcal{B} = [0,1]^{B+1} {\subset \mathbb{R}^{B+1}}}$ is the belief space. Formally, the belief $\beta(t)$ determines the conditional probability distribution that the battery level has a specific value at slot $t$, given the complete information state $\phi^\mathrm{c}(t)$. Accordingly, the entries of $\beta(t)$ are defined as 
	\begin{equation}\label{eq_def_update}
		\beta_{j}(t) \triangleq \Pr(b(t) = j \mid \phi^\mathrm{c}(t)), ~j \in \{0,1,\dots,B\}.
	\end{equation}

	The belief is updated at each slot based on the previous belief, the current observation, and the previous action. That is, ${\beta(t+1) = \tau(\beta(t),o(t+1),a(t))}$, where the belief update function $\tau(\cdot)$ is {given by the following theorem}.
	\begin{theorem}\label{lemma_beleief_update}
		The belief update function $\tau(\cdot)$ is given by
		\begin{equation}\label{eq_belief_update}
			\beta(t+1) =  \tau(\beta(t),o(t+1),a(t)) =  \left\lbrace 
			\begin{array}{ll}
				\boldsymbol{\Lambda}\beta(t), \hspace{0.0mm} & a(t) = 0,\\
				{\rho}^{0},\hspace{0.0mm}
				&a(t) = 1, \Delta(t+1) > 1,\\
				{\rho}^{1},\hspace{0.0mm} &a(t) = 1, \Delta(t+1) =1,\tilde{b}(t+1)=1,\\
				\dots\\
				{\rho}^{B},\hspace{0.0mm} &a(t) = 1,\Delta(t+1) =1,\tilde{b}(t+1)=B,
			\end{array}
			\right.
		\end{equation}
		where {$\Delta(t+1)$ and $\tilde{b}(t+1)$ are entries of $o(t+1)$,} the matrix $\boldsymbol{\Lambda} \in [0,~ 1]^{(B+1)\times(B+1)}$ is a left stochastic matrix, i.e., $\mathbf{1}^{\tran} \boldsymbol{\Lambda} = \mathbf{1}^{\tran}$, having a banded form as
		\begin{equation}\label{eq_matrix_lambda}
			\boldsymbol{\Lambda} = 
			\begin{pmatrix}
				1-\lambda & 0 & \cdots & 0 & 0 & 0 \\
				\lambda & 1-\lambda & \cdots & 0 & 0 & 0\\
				\vdots  & \vdots  & \ddots & \vdots & \vdots  \\
				0 & 0 & \cdots & \lambda & 1-\lambda & 0\\
				0 & 0 & \cdots & 0 & \lambda & 1
			\end{pmatrix},
		\end{equation}
		and the vectors ${{\rho}^{j}\in[0,~1]^{B+1}}$, ${j=0,1,\ldots,B}$, are given by
		\begin{equation}\label{eq_vectors_rho}
			\begin{array}{ll}
				& {{\rho}^{0}} = {\rho}^{1} = \big(
				1-\lambda,\lambda,0,0,\ldots,0,0\big) \tran \\
				&{\rho}^{2} = 
				\big(0,1-\lambda,\lambda,0,0,\ldots,0,0\big) \tran \\
				& \vdots\\
				&{\rho}^{B} = 
				\big(0,0,\ldots,0,0,1-\lambda,\lambda\big) \tran.
			\end{array}
		\end{equation}
	\end{theorem}
	\begin{proof}
		The details of the proof are presented in Appendix~\ref{sec_appendix_lemma_beleief_update}. 
		Intuitively, when $a(t) = 0$, the edge node does not receive an update,
		and thus, the belief is updated based on the previous belief and the fact that the energy arrivals are modeled as independent Bernoulli variables. To exemplify, the probability that $b(t+1) = 0$ (i.e., $\beta_0(t+1)$) is the product of two independent probabilities: the probability that the battery level was zero at $t$ (i.e., $\beta_0(t)$) and the sensor did not receive one unit of energy during slot $t$ (i.e., $1-\lambda$). Thus, $\beta_0(t+1) = \Pr(e(t) = 0)\beta_0(t) = (1-\lambda) \beta_0(t)$ (see the first row of the matrix $\boldsymbol{\Lambda}$). By the similar logic, $\beta_1(t+1) = \Pr(e(t) = 1)\beta_0(t) + \Pr(e(t) = 0) \beta_1(t) = \lambda\beta_0(t) + (1-\lambda) \beta_1(t)$ (see the second row of  $\boldsymbol{\Lambda}$), and etc.
		For the case where $a(t) = 1$, if the edge node does not receive an update (i.e., $\Delta(t+1) > 1$), then it is inferred that $b(t) = 0$. Thus, we either have $b(t+1) = 0$ or $b(t+1) = 1$, which happens with probability $1-\lambda$ and $\lambda$, respectively (see ${{\rho}^{0}}$). 
		For the case where $a(t) = 1$ and the edge node receives an update (i.e., $\Delta(t+1) = 1$), the edge node also receives $b(t)$ as part of the update packet ($b(t) = i \geq 1$). Besides, note that one unit of energy has been consumed to send the update. Therefore, the belief about the battery level at $t+1$ is $1-\lambda$ and $\lambda$ at entries $i-1$ and $i$ (see ${{\rho}^{1}}$, $\dots$, ${{\rho}^{B}}$).
	\end{proof}

	

	\subsection{Optimal Policy and Proposed Algorithm}\label{sec-per-senspr-optimal-policy}
	In this section, we find an optimal policy for the POMDP. Formally, a policy $\pi$ 
	decides which action $a$ to take at a particular belief-state $z$. The policy $\pi$ is either randomized or deterministic. 
	A randomized policy is determined by a  distribution ${\pi(a \mid z) : \mathcal{Z} \times \mathcal{A} \rightarrow \left[ 0 , 1\right]}$, whereas a deterministic policy is determined by a mapping ${\pi: \mathcal{Z} \rightarrow \mathcal{A}}$.
	For deterministic policies, we use $\pi(z)$ to denote the action taken in belief-state $z$  by a deterministic policy $\pi$.
	Under a policy $\pi$, the average cost is given by (see \eqref{eq_average_cost_persensor})
	\begin{equation}\label{eq_average_cost_policy}
		\bar{C}_{\pi}=\lim_{T\rightarrow\infty} \frac{1}{T}\sum_{t=1}^{T} \mathbb{E}_{\pi}[\Delta^{\mathrm{OD}}(t) \mid z(0)],
	\end{equation}
	where $z(0)$ is the initial belief-state\footnote{We assume that all policies $\pi$  induce a Markov chain with a single recurrent class plus a (possibly empty) set of transient states (i.e., the uni-chain condition is satisfied). Consequently, the minimum average cost  does not depend on the initial state \cite[Chapter~8]{puterman2014markov}. It is worth noting that, in general, checking the uni-chain condition for an MDP is NP-Hard \cite{TSITSIKLIS2007NPhardUnichain}. Importantly, 
		{the} 
		assumption makes problem \eqref{eq_average_cost_persensor_pomdp} well-posed so that we can use the tools associated with the uni-chain MDPs.}.
	{Note that $\mathbb{E}_{\pi}[\cdot]$
		denotes the expected value of the on-demand AoI
		when the policy $\pi$ is employed.}
	We aim to find an optimal policy 
	that minimizes \eqref{eq_average_cost_policy}, i.e., 
	\begin{equation}\label{eq_average_cost_persensor_pomdp}
		\pi^* \in \argmin_{\pi} \bar{C}_{\pi},
	\end{equation}
	where the minimization is with respect to all deterministic or randomized policies.
	The following theorem characterizes an optimal  policy $\pi^*$. 
	\begin{theorem}\label{theorem_bellman_eq}
		An optimal policy $\pi^*$ is obtained by solving the
		following equations:
		\begin{equation}\label{eq_ballman_pomdp}
			\bar{C}^* + h(z) = \min_{a \in \mathcal{A}} Q(z,a),~ z \in \mathcal{Z},
		\end{equation}
		where $h(z)$ is a relative value function, $\bar{C}^*$ is the optimal average cost achieved by $\pi^*$ which is independent of the initial state $z(0)$, and $Q(z, a)$ is an action-value function, which for belief-state ${z = (\beta,r,\Delta,\tilde{b}) \in \mathcal{Z}}$ and action ${a \in \{0,1\}}$, is given by  
		\begin{subequations}\label{eq_bellman_q_pomdp}
			\begin{align}
				Q(z , 0)  &= r \min\{\Delta+1,\Delta^{\mathrm{max}}\}+ \notag \\
				& \hspace{5mm} \sum_{r^\prime = 0}^{1} [r^\prime p + 
				(1-r^\prime)(1-p)] h(\boldsymbol{\Lambda} \beta,r^\prime,\min\{\Delta+1,\Delta^{\mathrm{max}}\}, \tilde{b}), \\
				Q(z ,1)  &= r[\beta_{0} \min\{\Delta+1,\Delta^{\mathrm{max}}\} + (1-\beta_{0})] + \beta_0 \sum_{r^\prime = 0}^{1} [r^\prime p + (1-r^\prime)(1-p)]\notag \\
				&h({\rho}^{0},r^\prime,\min\{\Delta+1,\Delta^{\mathrm{max}}\},\tilde{b})  + \sum_{j = 1}^{B} \beta_j \big[p h({\rho}^{j},1,1,j) + (1-p) h({\rho}^{j},0,1,j) \big].
			\end{align}
		\end{subequations}
		Further, an optimal action taken 
		in belief-state $z$ is obtained as 
		\begin{equation}\label{eq_optimal_policy}
			\pi^*(z) \in \argmin_{a\in\mathcal{A}} Q(z,a),~ z\in \mathcal{Z}.  
		\end{equation}
	\end{theorem}
	\begin{proof}
		{See Appendix~\ref{sec_appendix_theorem_bellman_eq}.}
	\end{proof}
	

	Bellman's optimality equation \eqref{eq_bellman_q_pomdp} can be solved iteratively through a method called relative value iteration algorithm (RVIA) \cite[Section~8.5.5]{puterman2014markov}. Specifically, at each iteration ${i=1,2,\ldots}$, 
	we update {$Q^{(i)}(z,a)$ in \eqref{eq_bellman_q_pomdp} by using $h^{(i-1)}(z)$, {leading to the following updates:}
		\begin{equation}\label{eq_v_itr}
			\begin{aligned}
				&V^{(i)}(z) =\min_{a\in\mathcal{A}}Q^{(i)}(z,a), ~ z \in \mathcal{Z}\\
				& h^{(i)}(z) = V^{(i)}(z) - V^{(i)}(z_{\mathrm{ref}}),
			\end{aligned}
	\end{equation}}
	where ${z_{\mathrm{ref}} \in \mathcal{Z}}$ is an {arbitrarily chosen} reference belief-state.
	Regardless of $V^{(0)}(z)$, the sequences $\{Q^{(i)}(z,a)\}_{i=1,2,\ldots}$, $\{h^{(i)}(z)\}_{i=1,2,\ldots}$, and $\{V^{(i)}(z)\}_{i=1,2,\ldots}$ converge \cite[Section~8.5.5]{puterman2014markov}, i.e., ${\lim_{i\to\infty}Q^{(i)}(z,a) = Q(z,a)}$, ${\lim_{i\to\infty}h^{(i)}(z) = h(z)}$, and ${\lim_{i\to\infty}V^{(i)}(z) = V(z)}$, $\forall z$.
	Thus, $h(z) = V(z) - V(z_{\mathrm{ref}})$ satisfies \eqref{eq_ballman_pomdp} 
	and $\bar{C}^* = V(z_{\mathrm{ref}})$.

	Although the sequences in \eqref{eq_v_itr} converges, finding $V(z)$ (and $h(z)$) iteratively via \eqref{eq_v_itr} is intractable, because the belief space $\mathcal{B}$ has infinite dimension. Fortunately, the evolution of the beliefs $\{\beta(t)\}_{t = 0,1,\dots}$ has a certain pattern which we exploit to truncate the belief space $\mathcal{B}$ and subsequently, to develop a practical iterative algorithm relying on \eqref{eq_v_itr}.
	To illustrate the pattern, consider an initial belief $\beta(0) = \beta$. Then, by Theorem~\ref{lemma_beleief_update},
	when action ${a = 0}$ is taken, the next belief is ${\beta(1) = \boldsymbol{\Lambda} \beta}$ and when action ${a=1}$ is taken, the next belief is ${\beta(1) \in \{{\rho}^{1},\dots,{\rho}^{B}\}}$. {Thus}, the belief at slot ${t = 1}$ belongs to the set $\{\boldsymbol{\Lambda} \beta,\{{\rho}^{j}\}_{j=1}^{B}\}$. Similarly, the belief at ${t = 2}$ belongs to ${\{\boldsymbol{\Lambda}^2 \beta,\{\boldsymbol{\Lambda}{\rho}^{j}\}_{j=1}^{B},\{{\rho}^{j}\}_{j=1}^{B}\}}$, the belief at ${t = 3}$ belongs to  ${\{\boldsymbol{\Lambda}^3 \beta,\{\boldsymbol{\Lambda}^{2}{\rho}^{j}\}_{j=1}^{B},\{\boldsymbol{\Lambda}{\rho}^{j}\}_{j=1}^{B},\{{\rho}^{j}\}_{j=1}^{B}\}}$, and so on. This pattern in the belief evolution is depicted in Fig.~\ref{fig_belief_Transistion}.
	Accordingly, the belief space $\mathcal{B}$ containing all the possible beliefs $\beta(t)$, $\forall t$, is infinite but \textit{countable} given initial belief $\beta(0) {= \beta}$ and energy harvesting rate $\lambda$, as shown in Table \ref{tab:belief_set}.
	

	\begin{table}[t]
		\centering
		\caption{
			Illustration of the belief space $\mathcal{B}$ and the truncated belief-space $\hat{\mathcal{B}}$ given an initial belief $\beta(0) = \beta$. In the simulation results, the \textit{row} and \textit{column} numbers are used to represent each belief, e.g., $\boldsymbol{\Lambda}^2\boldsymbol{\rho^{1}}$ is shown by $(1,2)$.}
		\label{tab:belief_set}
		\scalebox{0.9}{
			\begin{adjustbox}{width=\columnwidth,center}
				\begin{tabular}{c:c  c  c  c c c c c}
					&\gray{0}&\gray{1}&\gray{2}&\gray{3}&$\dots$&$\gray{M}$&$\gray{M+1}$&\dots \\
					\hdashline
					\noalign{\vskip 0.08cm}
					\gray{0}&$\beta$ & $\boldsymbol{\Lambda}\beta$ & $\boldsymbol{\Lambda}^2\beta$ & $\boldsymbol{\Lambda}^3\beta$ & $\dots$ & $\boldsymbol{\Lambda}^{M}\beta$ & $\boldsymbol{\Lambda}^{M+1}\beta$ & $\dots$ \\
					\gray{1}&${\rho}^{1}$ & $\boldsymbol{\Lambda}{\rho}^{1}$ & $\boldsymbol{\Lambda}^{2}{\rho}^{1}$ & $\boldsymbol{\Lambda}^{3}{\rho}^{1}$ & $\dots$ & $\boldsymbol{\Lambda}^{M}{\rho}^{1}$ & $\boldsymbol{\Lambda}^{M+1}{\rho}^{1}$ & $\dots$\\
					$\vdots$&$\vdots$ & $\vdots$ & $\vdots$ & $\vdots$ & $\ddots$ & $\vdots$ & $\vdots$ &  $\dots$ \\
					$\gray{B}$&${\rho}^{B}$ & $\boldsymbol{\Lambda}{\rho}^{B}$ & $\boldsymbol{\Lambda}^{2}{\rho}^{B}$ & $\boldsymbol{\Lambda}^{3}{\rho}^{B}$ & $\dots$ & $\boldsymbol{\Lambda}^{M}{\rho}^{B}$& $\boldsymbol{\Lambda}^{M+1}{\rho}^{B}$ & $\dots$\\
					\noalign{\vskip -0.20cm}
					&\multicolumn{6}{c}{\upbracefill}\\
					\noalign{\vskip -0.05cm}
					&\multicolumn{6}{c}{$\scriptstyle \hat{\mathcal{B}}$}\\
				\end{tabular}
			\end{adjustbox}
		}
	\end{table}
	
	\begin{figure}[t!]
		\centering
		{\includegraphics[width=0.6\columnwidth]{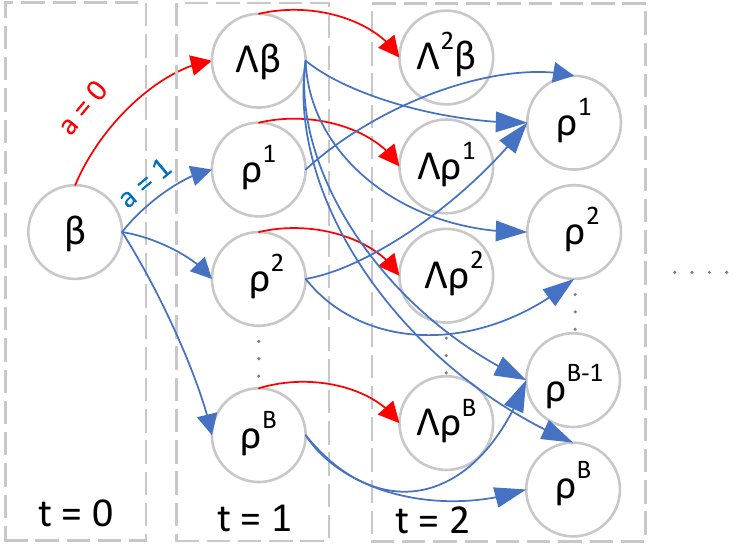}}
		\caption{Evolution of the belief $\beta(t)$ over time. {Red arrow: no command $a=0$; blue arrow: command $a=1$.}}
		\label{fig_belief_Transistion}
	\end{figure}

	The following lemma reveals the key property of matrix $\boldsymbol{\Lambda}$ in \eqref{eq_matrix_lambda}, which will be used to truncate the belief space $\mathcal{B}$ into a finite space. 
	\begin{lemm}\label{lemm_struct_LAMBDA}
		The $m$th power of matrix $\boldsymbol{\Lambda}$ is given by
		\begin{align}\label{eq_Lambda_power_m}
			\Lambda_{j,l}^m  =  \left \lbrace
			\begin{array}{ll}
				0,   & j < l, \\
				(1-\lambda)^m,  & j = l,\\
				\lambda^{(j-l)}(1-\lambda)^{(m-j+l)}\prod_{v = 0}^{j-l-1}\frac{(m-v)}{(v+1)},  & l < j \leq B,\\
				1 - \sum_{j^\prime = 1}^{B}\Lambda_{j^\prime,l}^m,  & j = B + 1, \forall l,
			\end{array}
			\right.
		\end{align}
		where $\Lambda_{j,l}^m$ is the {entry} of matrix $\boldsymbol{\Lambda}^m$ at its $j$th row and $l$th column.
	\end{lemm}
	\begin{proof}
		See Appendix~\ref{sec_appendix_lemm_struct_LAMBDA}.
	\end{proof}
	
	Thus, by \eqref{eq_Lambda_power_m}, for any energy arrival rate ${0 < \lambda \leq 1}$, we have
	\begin{equation}
		\lim_{m \rightarrow \infty} \Lambda_{j,l}^m = \left \lbrace
		\begin{array}{ll}
			0,   & j \leq B, \forall l,\\
			1,  & j = B + 1, \forall l,
		\end{array}
		\right.
	\end{equation}
	{which states that when ${m \rightarrow \infty}$, matrix $\boldsymbol{\Lambda}^{m}$ tends to a matrix with all entries zero except that the entries at its last row are all ones. Consequently, $\lim_{m \rightarrow \infty}\boldsymbol{\Lambda}^m\beta \rightarrow (0,0,\dots,0,1)\tran$}, $\forall \beta$, and, for a sufficiently large integer $M$, we have {$\boldsymbol{\Lambda}^{M} \approx \boldsymbol{\Lambda}^{M+1}$}. Thus, we construct a truncated belief space $\hat{\mathcal{B}}$ of finite dimension $|\hat{\mathcal{B}}| = (B+1)(M+1)$, as shown in Table \ref{tab:belief_set}. Intuitively, the value $M$ represents the maximum number of consecutive no-command actions ($a = 0$) for which the belief is updated. This means that from the $(M+1)$th no-command onward, the belief is no longer updated.
	{This is reasonable because after $M$ consecutive $a = 0$ actions, the battery is almost full,
		i.e., $\boldsymbol{\Lambda}^M \beta \approx (0,0,\dots,0,1)\tran$,~$\forall \beta \in \mathcal{B}$, and thus, for sufficiently large $M$, the space $\hat{\mathcal{B}}$ covers (almost) all the possible beliefs. In other words, $\hat{\mathcal{B}}$ includes all possible beliefs 
		for any number of time steps with maximum $M$ consecutive zero actions.}
	{The convergence rate of the power of the matrix $\boldsymbol{\Lambda}$ (i.e., $\boldsymbol{\Lambda}^m$) is determined by the second largest eigenvalue modulus (SLEM) of $\boldsymbol{\Lambda}$.
		Specifically, the smaller the SLEM, the faster the  $\boldsymbol{\Lambda}^m$ converges \cite[Section~1.1.2]{boyd2004fastest}.  
		The eigenvalues of $\boldsymbol{\Lambda}$ are $1$ and $1-\lambda$, and hence,  $\boldsymbol{\Lambda}^m$ converges faster
		as $\lambda$ increases.}

	\begin{algorithm}[t]
		\caption{Proposed algorithm to obtain an optimal policy $\pi^*$
		}\label{alg_value_itr_pomdp}
		\begin{algorithmic}[1]
			\STATE \textbf{Initialize} $V(z) = h(z) =0$,$\forall z = (\beta,r,\Delta)$, $\beta\in \hat{\mathcal{B}}$, $r \in  \{0,1\}$, $\Delta \in \{1,\dots,\Delta^{\mathrm{max}}\}$, determine an arbitrary $z_{\textrm{ref}} \in \mathcal{Z}$ and a small threshold $\theta > 0$
			\REPEAT
			\FOR{$z$}
			\STATE calculate $Q(z,0)$ and $Q(z,1)$ by using \eqref{eq_bellman_q_pomdp}
			\STATE $V_{\mathrm{tmp}}(z) \leftarrow \min_{a \in \mathcal{A}} Q(z,a)$
			\ENDFOR
			\STATE $\delta \leftarrow \max_{z} (V_{\textrm{tmp}}(z) - V(z)) - \min_{z} (V_{\textrm{tmp}}(z) - V(z))$
			\STATE $V(z) \leftarrow V_{\textrm{tmp}}(z)$,  for all $z$ 
			\STATE $h(z) \leftarrow V(z) - V(z_{\mathrm{ref}})$,  for all $z$
			\UNTIL{$\delta < \theta$}
			\STATE $\pi^*(z) = \argmin_{a \in \mathcal{A}} Q(z,a)$, for all $z$
		\end{algorithmic}
	\end{algorithm}

	Next, we provide a theorem which is used to reduce the size of the belief-state space, thereby leading {to a reduced} complexity of the proposed algorithm.
	
	\begin{theorem}\label{theorem_structure_v_part1}
		Function ${V(z)}$ associated with a belief-state ${z=(\beta,r,\Delta,\tilde{b})}$ does not depend on partial battery knowledge $\tilde{b}$.
	\end{theorem}
	\begin{proof}
		See Appendix~\ref{sec_appendix_theorem_structure_v_part1}.
		{Intuitively, the belief $\beta$ must capture all
			the {relevant}
			information in $\tilde{b}$ regarding searching for an optimal policy, and {consequently}, there is not any extra information in $\tilde{b}$ given $\beta$.}
	\end{proof}

	
	{According to} Theorem~\ref{theorem_structure_v_part1}, $V(z)$, where $z = (\beta,r, \Delta, \tilde{b})$, {and consequently $h(z)$ and $Q(z,a)$}, do not depend on $\tilde{b}$. Thus, $\tilde{b}$ does not have any impact on calculating $\pi^*$ in \eqref{eq_optimal_policy}. Therefore, we remove\footnote{Note that while $\tilde{b}$ is removed from the belief-state, it is still needed to calculate the belief $\beta(t)$. 
	} $\tilde{b}$ from the belief-state $z$ and redefine the belief-state {$z$ (and also the belief-state space $\mathcal{Z}$)} hereinafter as $z = (\beta,r, \Delta) \in \mathcal{Z}$.
	We note that \eqref{eq_bellman_q_pomdp} can easily be rewritten based on the new belief-state definition {by dropping the last entry in $h$.}

	Finally, considering the truncated belief space $\hat{\mathcal{B}}$, we use \eqref{eq_bellman_q_pomdp}--\eqref{eq_v_itr} to find $V(z)$, $h(z)$, $Q(z,a)$, and consequently an optimal policy $\pi^*$ iteratively, as presented in Algorithm~\ref{alg_value_itr_pomdp}. 

	\subsection{Efficient Algorithm Implementation using the Sparsity of Transition Matrices}
	We assign an index $z = 1,2, \dots, |\mathcal{Z}|$ to each belief-state $z = (\beta,r,\Delta)$, $\beta \in \hat{\mathcal{B}}$, $r \in \{0,1\}$, $\Delta \in \{1,2,\cdots,\Delta^{\mathrm{max}}\}$. 
	{We define the cost vector associated with action $a \in \{0,1\}$ as $\mathbf{c}^a \triangleq (c(1,a),c(2,a),\dots,c(|\mathcal{Z}|,a))^\mathrm{T}$, where $c(z,a) \triangleq \sum_{b = 0}^B \beta_b c(s = ( b ,r,\Delta),a)$ is the immediate cost of taking action $a$ in belief-state $z$. We also define the belief-state transition matrix associated with action $a$ as $\mathbf{P}^a$, where $P^a_{j,l} \triangleq \Pr(z^\prime = l \mid z = j, a)$ is the entry of $\mathbf{P}^a$ at the $j$th row and $l$th column.} 
	Therefore, \eqref{eq_v_itr} can be written in the vector form as  
	\begin{equation}\label{vector-form-Bellman}
		\begin{aligned}
			& \mathbf{v}^{(i)} = \min_{a}\left[ \mathbf{c}^a + \mathbf{P}^a\mathbf{h}^{(i-1)}\right]
			\\& \mathbf{h}^{(i)} = \mathbf{v}^{(i)} - V^{(i)}(z_{\mathrm{ref}})\mathbf{1}, 
		\end{aligned}
	\end{equation}
	where $\mathbf{v}^{(i)} \triangleq (V^{(i)}(1),V^{(i)}(2),\dots,V^{(i)}(|\mathcal{Z}|))^\mathrm{T}$ and $\mathbf{h}^{(i)} \triangleq (h^{(i)}(1),h^{(i)}(2),\dots,h^{(i)}(|\mathcal{Z}|))^\mathrm{T}$ are column vectors, and $\mathbf{1}$ is a column vector with all entries $1$. Note that the operator $\min$
	denotes the element-wise minimum operation on two vectors.
	
	{The vector form \eqref{vector-form-Bellman} allows a more efficient implementation of Algorithm~\ref{alg_value_itr_pomdp}, as it eliminates the need for the for-loop over all belief-states in each iteration.} 
	The resulting algorithm is presented in Algorithm~\ref{alg_value_itr_pomdp_vector_form}, where the span of a vector $\mathbf{v}$ is defined as $\mathrm{sp}(\mathbf{v}) \triangleq \max_{z} V(z) - \min_{z} V(z)$, and an optimal policy vector is defined as $\boldsymbol{\pi}^* \triangleq (\pi^*(1),\dots, \pi^*(|\mathcal{Z}|))$.
	Note that the transition matrices $\mathbf{P}^{a}$, $a \in \{0,1\}$, are sparse, and this property is used to efficiently compute the matrix-vector products\footnote{{In particular, MATLAB, Python NumPy, Intel MKL, GNU Octave, and Julia have optimized routines for sparse matrix-vector multiplication that can handle large sparse matrices efficiently.}} (e.g., $\mathbf{P}^{a} \mathbf{h}^{(i)}$) using sparse matrix-vector multiplication methods. 
	\begin{Pro}\label{theorem-sparsity-transition-matrix}
		Denoting the number of nonzero elements in the sparse matrix $\mathbf{P}^{a}$ by  $\mathrm{nz}(\mathbf{P}^{a})$, we have $\mathrm{nz}(\mathbf{P}^{0}) = 2|\mathcal{Z}|$ and $\mathrm{nz}(\mathbf{P}^{1}) = 2(B+1)|\mathcal{Z}|$.
	\end{Pro}
	\begin{proof}
		See Appendix~\ref{sec-apndix-theorem-sparsity-transition-matrix}. {The sparsity structures will also be specified in the proof.}
	\end{proof}
	{The computational complexity of sparse matrix-vector multiplication is proportional to the number of nonzero elements in the matrix \cite[Appendix~C]{boyd2004convex}. Thus, by Proposition~\ref{theorem-sparsity-transition-matrix}, the computational complexity for each iteration of Algorithm~\ref{alg_value_itr_pomdp_vector_form} is ${O}(B|\mathcal{Z}|) = {O}(MB^2\Delta^{\mathrm{max}})$.}
	
	\begin{algorithm}[t!]
		\caption{Vector-based implementation of the proposed algorithm}\label{alg_value_itr_pomdp_vector_form}
		\begin{algorithmic}
			\STATE \textbf{Step 0.} {Initialize} $\mathbf{v}^{(0)} = \mathbf{h}^{(0)} = (0,\dots,0)$, set $i = 1$, and determine an arbitrary $z_{\textrm{ref}} \in \mathcal{Z}$, a small threshold $\theta > 0$
			\STATE \textbf{Step 1.} {Set}
			\begin{equation}\notag
				\begin{aligned}
					& \mathbf{v}^{(i)} = \min_{a}\left[ \mathbf{c}^a + \mathbf{P}^a\boldsymbol{h}^{(i-1)}\right]
					\\& \mathbf{h}^{(i)} = \mathbf{v}^{(i)} - V(z_{\mathrm{ref}})^{(i)}\mathbf{1}
				\end{aligned}
			\end{equation}
			\STATE \textbf{Step 2.} {If $\mathrm{sp}(\mathbf{v}^{(i)} - \mathbf{v}^{(i-1)}) < \theta$ , go to step 3; otherwise, $i \gets i+ 1$ and  go to step 1}
			\STATE \textbf{Step 3.} Compute optimal policy vector {$\boldsymbol{\pi}^* = \argmin_{a}\left[ \mathbf{c}^a + \mathbf{P}^a\mathbf{h}^{(i)}\right]$}
		\end{algorithmic}
	\end{algorithm}

	\subsection{Maximum Likelihood Estimator (MLE): a Sub-optimal MDP-Based Algorithm}\label{sec-suboptimal-MLE}
	Here, we propose a sub-optimal policy which has lower computational complexity than the optimal POMDP-based policy.
	{Assuming that we track the belief, a sub-optimal strategy is to act as if we were in the most likely state. Namely, the battery level with the highest probability mass is considered to be the battery level that the sensor is most likely to be in.}
	With this idea, we first consider the case where the edge node knows the exact battery levels at each time slot. In this case, an optimal policy, denoted by $\pi_{\mathrm{Exact}}^*(s)$, {$\forall s = (b,r,\Delta) \in \mathcal{S}$}, can be found {using relative value iteration algorithm (RVIA)} as shown in \cite{hatami2021spawc}.
	Then, we introduce the following sub-optimal policy for the case where the edge node does not know the exact battery level at each time slot. This sub-optimal policy is denoted by $\pi_{\mathrm{MLE}}$ and obtained by taking the learned $\pi_{\mathrm{Exact}}^*$, but then evaluating the policy by replacing the exact battery level $b$ with $b^* \triangleq \argmax_{j = 0,1,\dots,B} {\beta}_{j}$, i.e., $\pi_{\mathrm{MLE}}(z) =  \pi_{\mathrm{Exact}}^*(s^*)$, $s^* = (b^*,r,\Delta))$, $\forall z = (\beta,r,\Delta)$.
	Generally, the computational complexity for each iteration of the RVIA that finds $\pi_{\mathrm{Exact}}^*$ is  $O(|\mathcal{S}|) = {O}(B\Delta^{\mathrm{max}})$, whereas the computational complexity for each iteration of Algorithm~\ref{alg_value_itr_pomdp_vector_form} is ${O}(B|\mathcal{Z}|) = {O}(MB|\mathcal{S}|) = {O}(MB^2\Delta^{\mathrm{max}})$.

	\section{Multi-Sensor IoT Network: a Relax-then-Truncate Approach}\label{sec_multisensor-limited_bw}

	

	We extend the status update system to a multi-sensor IoT network under a transmission constraint. 
	We denote the set of sensors by $\mathcal{K} = \{1,2,\dots,K\}$, where $K$ is the number of sensors.
	Similarly to Section~\ref{sec_network}, we define different quantities associated with sensor $k$, e.g., the action of the edge node associated with sensor $k$ is denoted by $a_k(t) \in \{0,1\}$, $k \in \mathcal{K}$.
	We consider that, due to transmission limitations,  no more than $N$ sensors may send a status update packet to the edge node at each time slot. Thus, we have the following per-slot constraint
	\begin{equation}\label{eq_bw_st}
		\sum_{k = 1}^{K}a_k(t) \leq N, \forall t.
	\end{equation}
	
	For the case where $N \geq K$, the edge node can command any number of sensors at each slot, which implies the actions $a_k(t)$, $k\in\mathcal{K}$, are independent across $k$, and thus, the problem of finding an optimal policy reduces to finding per-sensor optimal policies individually.
	For $N < K$, we can model the problem as a POMDP and derive an optimal policy $\pi^*$ using a similar method as in Section~\ref{sec_single_sensor}. Particularly, the belief-state at slot $t$ is expressed as $\mathbf{z}(t) = (z_1(t),\dots,z_K(t))$, where the per-sensor belief-state $z_k(t)$, $k = 1,\dots,K$, were defined in Section~\ref{sec-belief-def}. 
	The edge node's action at slot $t$ is defined by a $K$-tuple $\mathbf{a}(t) = \big(a_1(t), \dots, a_K(t)\big) \in \mathcal{A}$, where the action space is $\mathcal{A}=\big\{(a_1,\ldots,a_K) \mid a_k\in \{0,1\}
	,\;\sum_{k=1}^{K}a_k\le{N}\big\}$. 
	It is worth noting that the belief-state space and action space grow exponentially with the number of sensors $K$, resulting exponential growth in the computational complexity of finding an optimal policy. 
	{Thus, inspired by \cite[Section~IV]{hatami2022JointTcom}, we next propose an asymptotically optimal low-complexity algorithm, called \textit{Relax-then-Truncate}, for which the  complexity grows only linearly with $K$.}
	

	We begin by relaxing the per-slot  constraint \eqref{eq_bw_st} into a time average constraint and model the \textit{relaxed problem} as a \textit{constrained POMDP} (CPOMDP).
	Leveraging the Lagrangian approach \cite{altman1999constrained}, we convert the CPOMDP into an unconstrained problem.
	The resulting POMDP decouples along the sensors, allowing us to find optimal per-sensor policies for a fixed Lagrange multiplier using the method described in Section~\ref{sec-per-senspr-optimal-policy}. We then determine the optimal Lagrange multiplier by applying the bisection method.
	This procedure provides an optimal policy for the relaxed problem, denoted by $\pi_{\mathrm{R}}^*$ and referred to as \textit{optimal relaxed policy} hereinafter. 
	Note that $\pi_{\mathrm{R}}^*$ may not satisfy the per-slot constraint \eqref{eq_bw_st}.
	{To ensure that \eqref{eq_bw_st} is satisfied at each slot, we {use} an online truncation procedure.}
	Specifically, at each slot, if the number of sensors commanded under $\pi_{\mathrm{R}}^*$ is less than or equal to $N$, all of those sensors are commanded, and if it is greater than $N$, a (uniformly) random subset of $N$ sensors is selected to be commanded.
	We elaborate the details in \cite{hatami2022POMDP_multisensor}
	and hence are omitted here for brevity. 
	{Our optimality analysis in \cite{hatami2022POMDP_multisensor} shows that the proposed {relax-then-truncate} is \textit{asymptotically optimal} as the number of sensors goes to infinity.} 
	


	\begin{theorem}\label{theorem-asymptotically-opt}
		For any normalized transmission budget $\Gamma \triangleq \frac{N}{K} > 0$, 
		The relax-then-truncate policy $\tilde{\pi}$ is asymptotically optimal with respect to the number of sensors, i.e., ${\lim_{K\rightarrow\infty} (\bar{C}_{\tilde{\pi}} - \bar{C}_{\pi^*}) = 0}$.
	\end{theorem}
	\begin{proof}
		The proof is presented in detail in \cite[Section~III-C]{hatami2022POMDP_multisensor}.
	\end{proof}
	
	
	
	
	\section{Simulation Results}\label{sec_simulation}
	In this section, we provide simulation results to demonstrate the performance of the proposed status update algorithms for both single-sensor and multi-sensor scenarios.
	
	\newcommand\figwidth{0.49}
	\begin{figure}[t]
		\centering
		\subfigure [$r = 0$]{%
			\includegraphics[width=\figwidth \columnwidth]{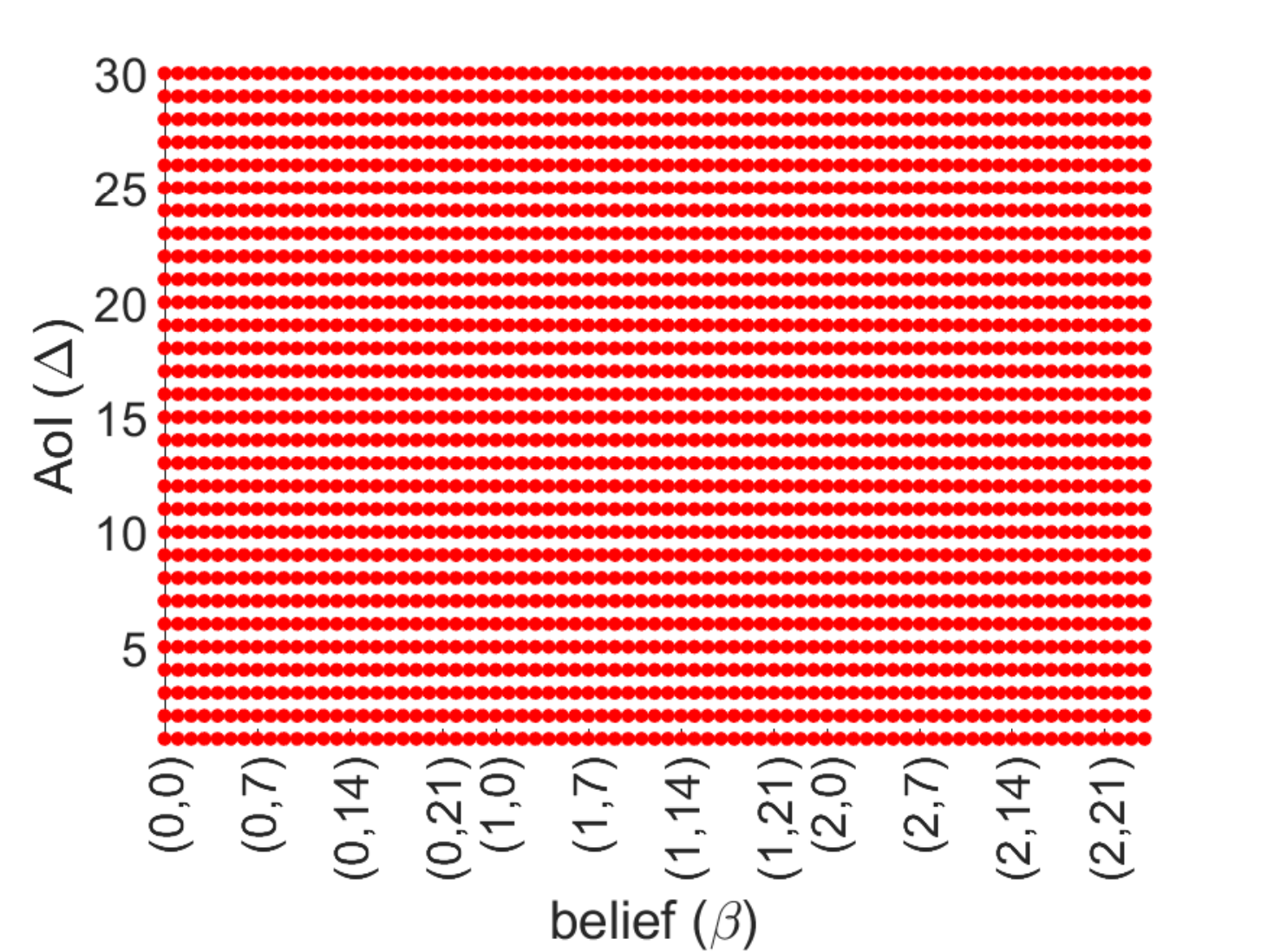}\label{fig_structure_a}
		}
		\subfigure [$r = 1$]{%
			\includegraphics[width=\figwidth \columnwidth]{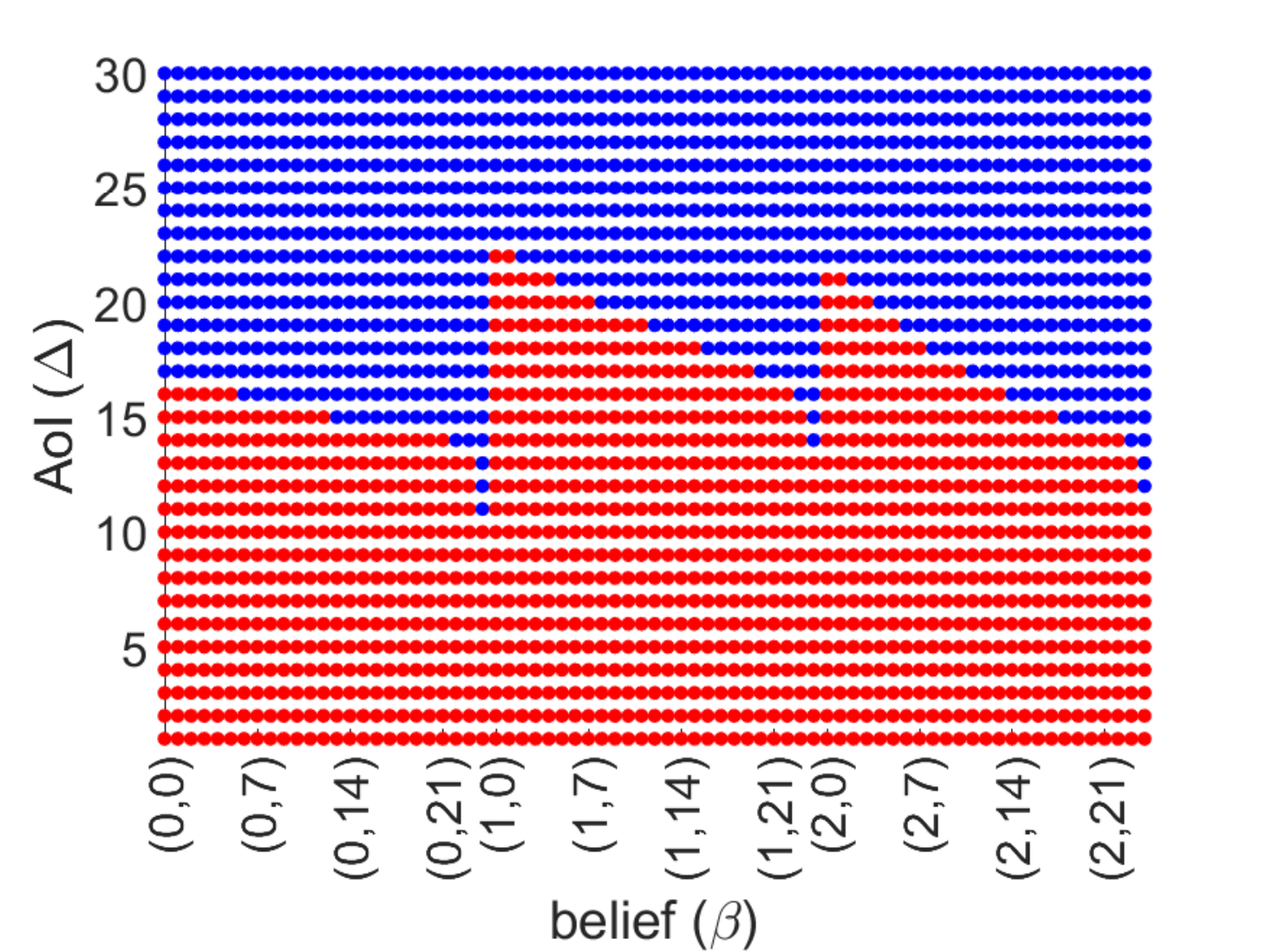}%
		}
		\caption{Structure of an optimal policy $\pi^*(z)$ for each belief-state $z = (\beta,r, {\Delta})$, where $p = 0.8$, $\lambda = 0.06$, and initial belief $\beta({0}) = (1/3,1/3,1/3)$.}\label{fig_structure}
		\subfigure [$p = 0.8$ and $\lambda = 0.12$]{%
			\includegraphics[width=\figwidth \columnwidth]{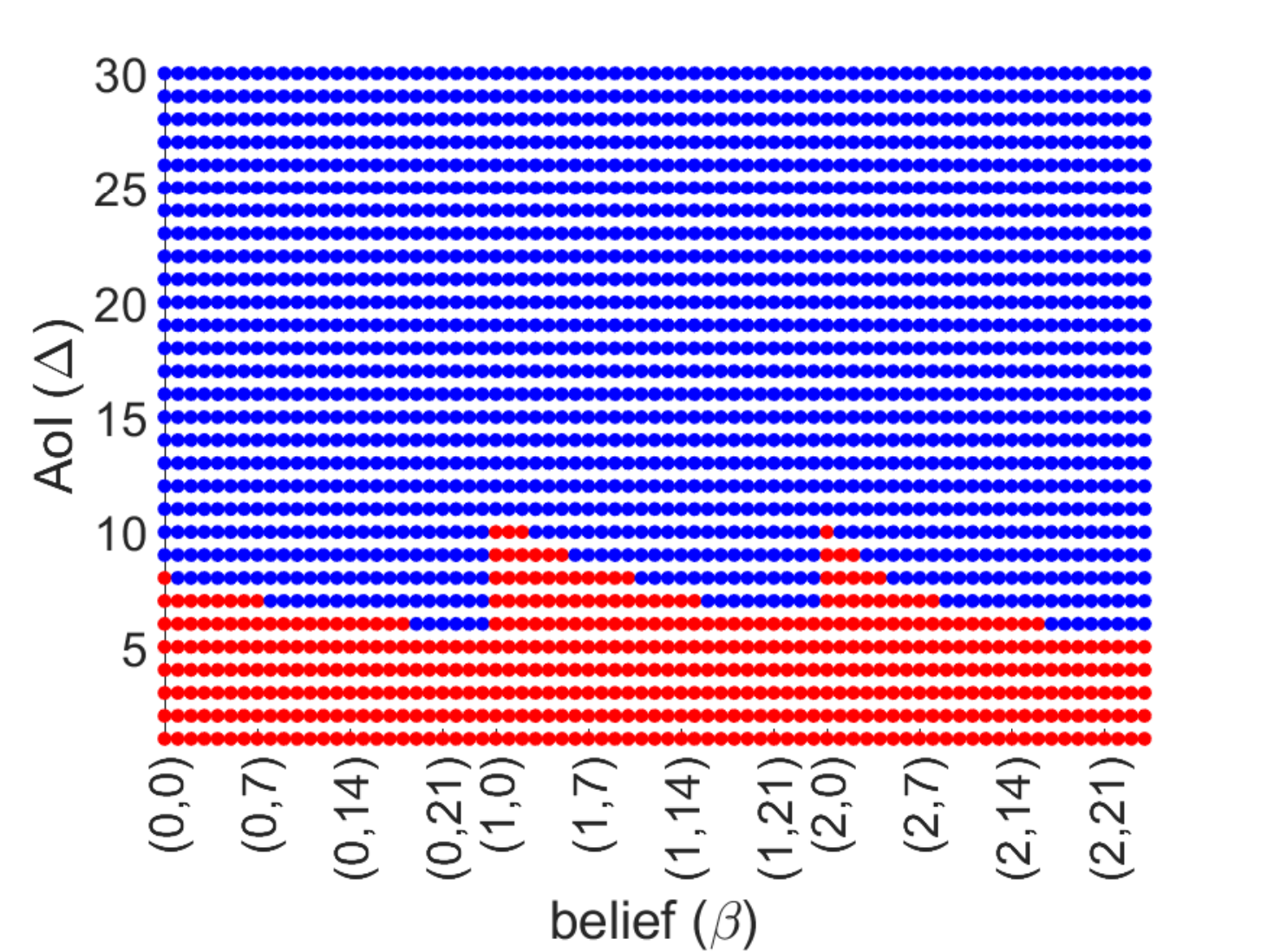}%
		}
		\subfigure [$p = 0.2$ and $\lambda = 0.06$]{%
			\includegraphics[width=\figwidth \columnwidth]{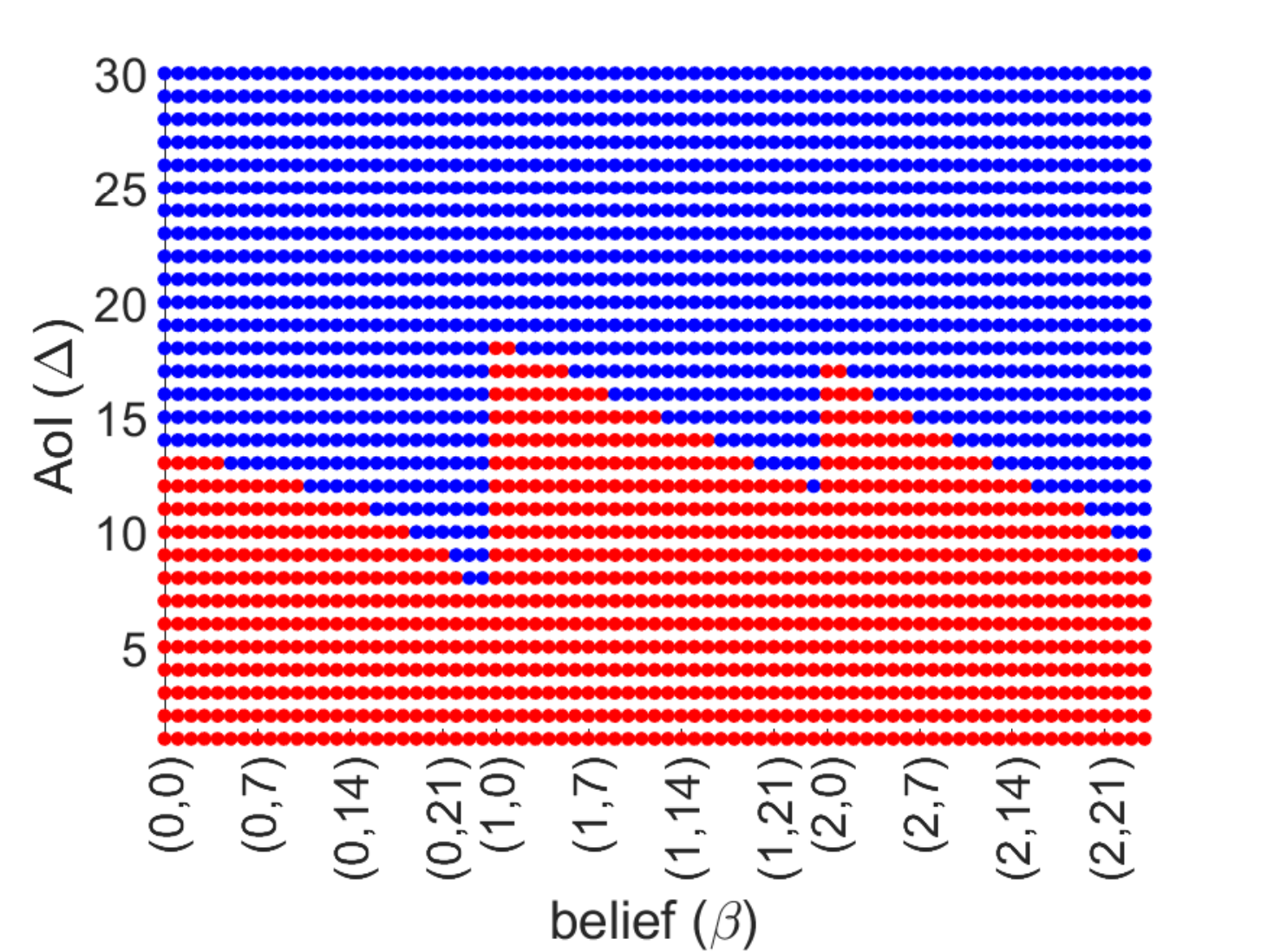}%
		}
		\caption{Structure of an optimal policy $\pi^*(z)$ for each belief-state $z = (\beta,1, {\Delta})$.}
		\label{fig_structure_compare}
	\end{figure}

	\newcommand\figPerfwidth{0.49}
	\begin{figure}[t]
		\centering
		\subfigure [$\lambda = 0.04$]{%
			\includegraphics[width=\figPerfwidth \columnwidth]{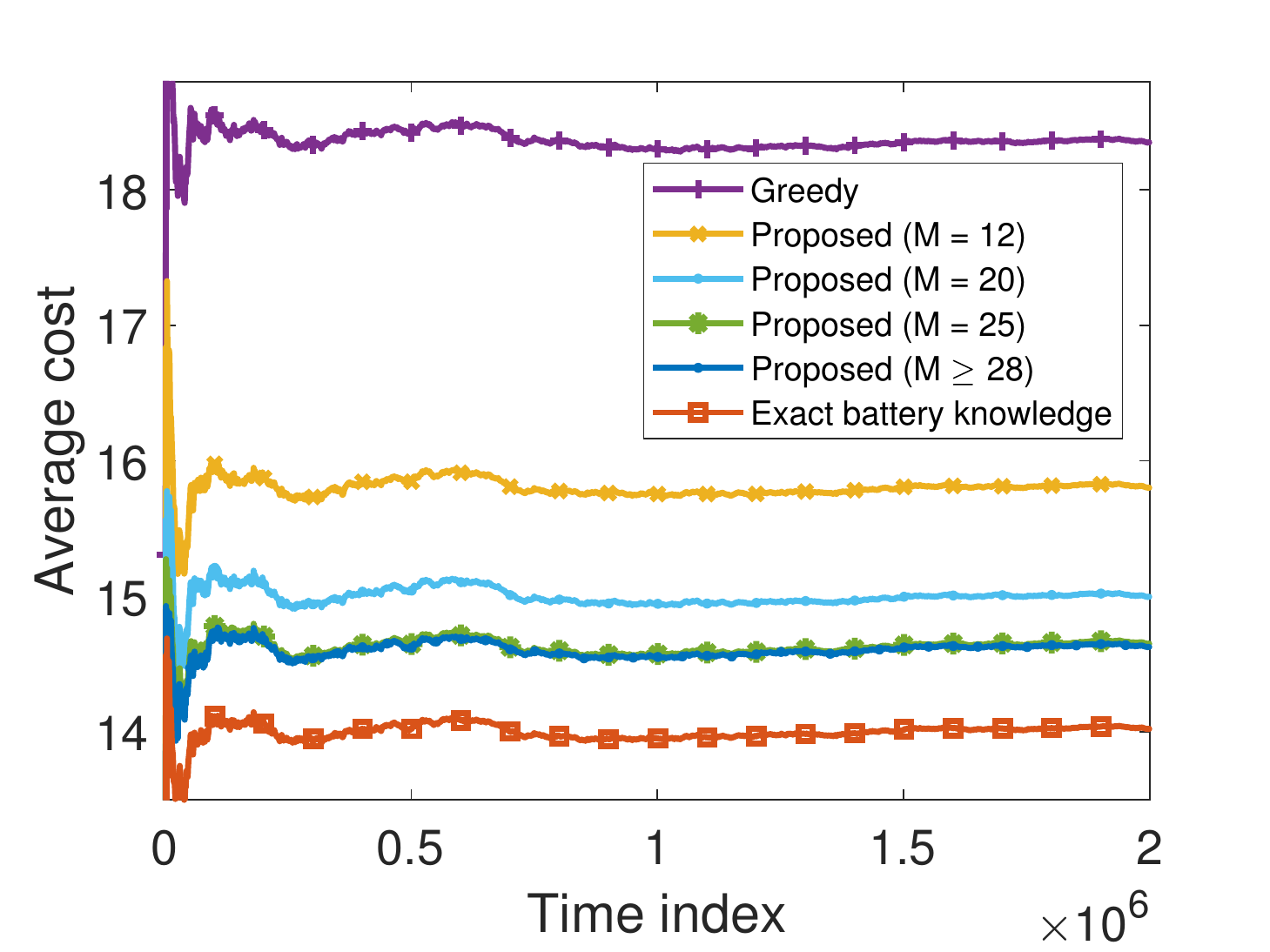}%
		}
		\subfigure [$\lambda = 0.08$]{%
			\includegraphics[width=\figPerfwidth \columnwidth]{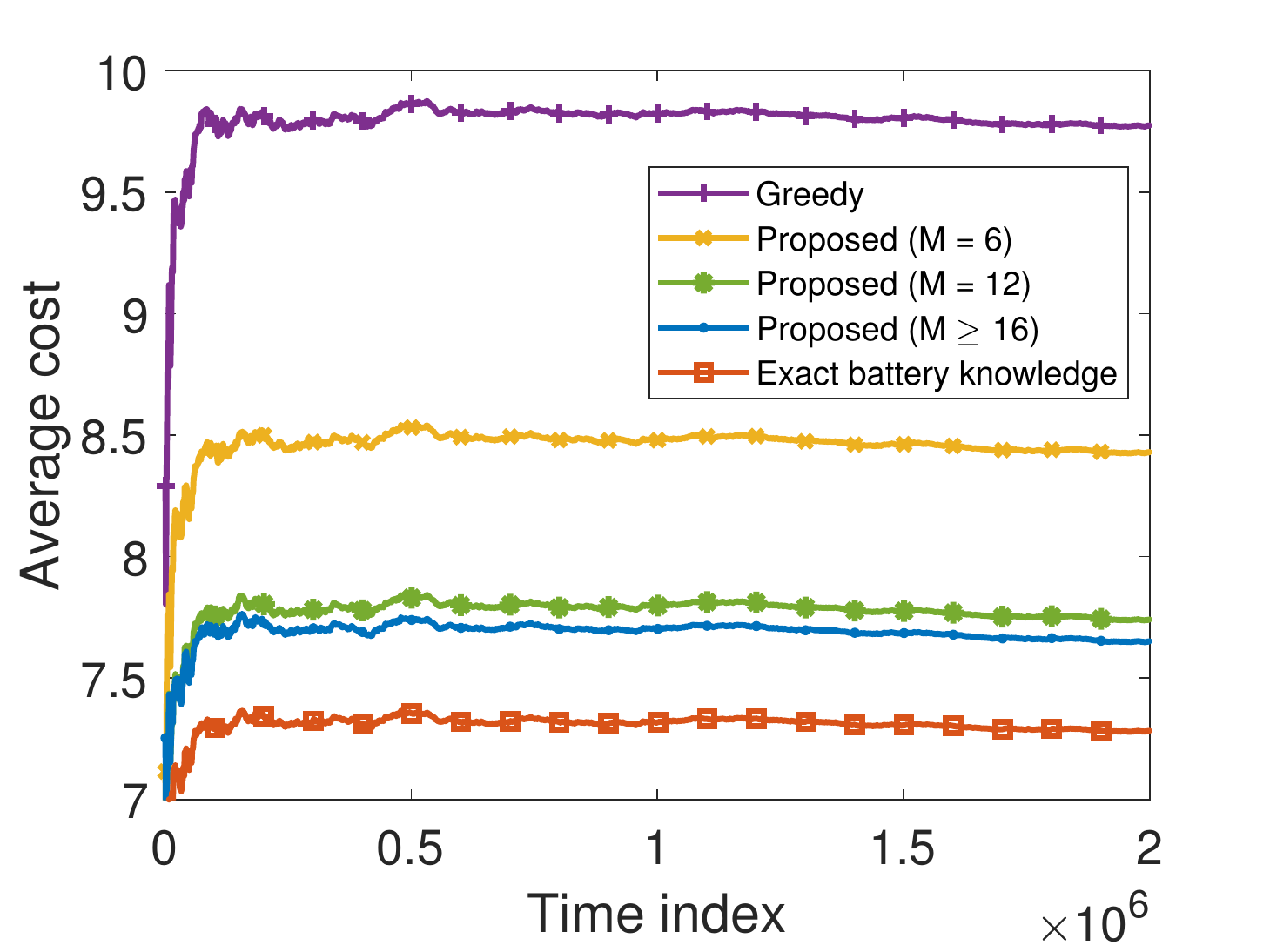}%
		}
		\caption{Performance of the proposed POMDP-based algorithm over time for single sensor scenario.}
		\label{fig_performance}
	\end{figure}
	
	\begin{figure}[t]
		\centering
		\subfigure [$p = 0.8$, $B = 2$]{%
			\includegraphics[width=\figPerfwidth \columnwidth]{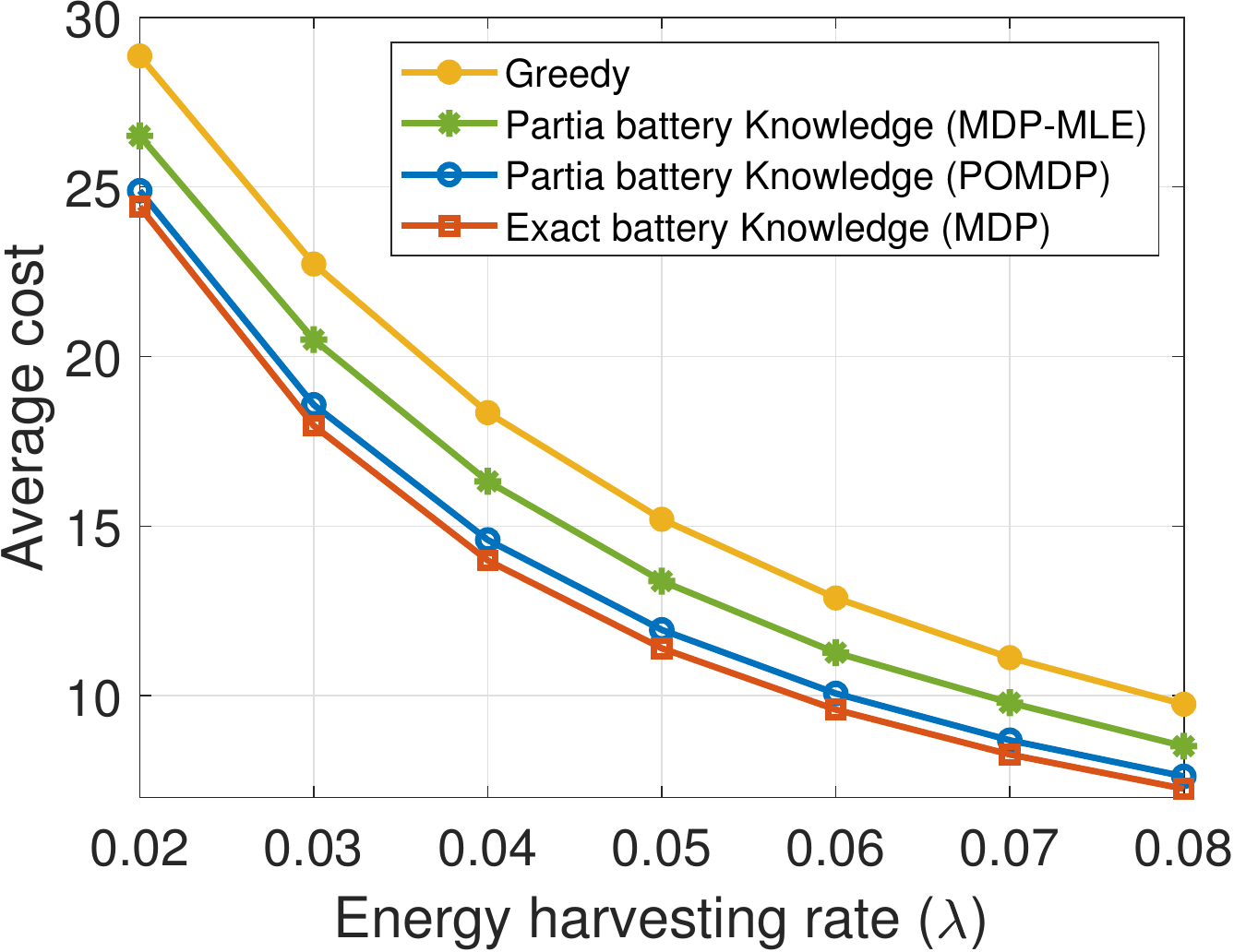}%
		}
		\subfigure [$p = 0.8$, $B = 3$]{%
			\includegraphics[width=\figPerfwidth \columnwidth]{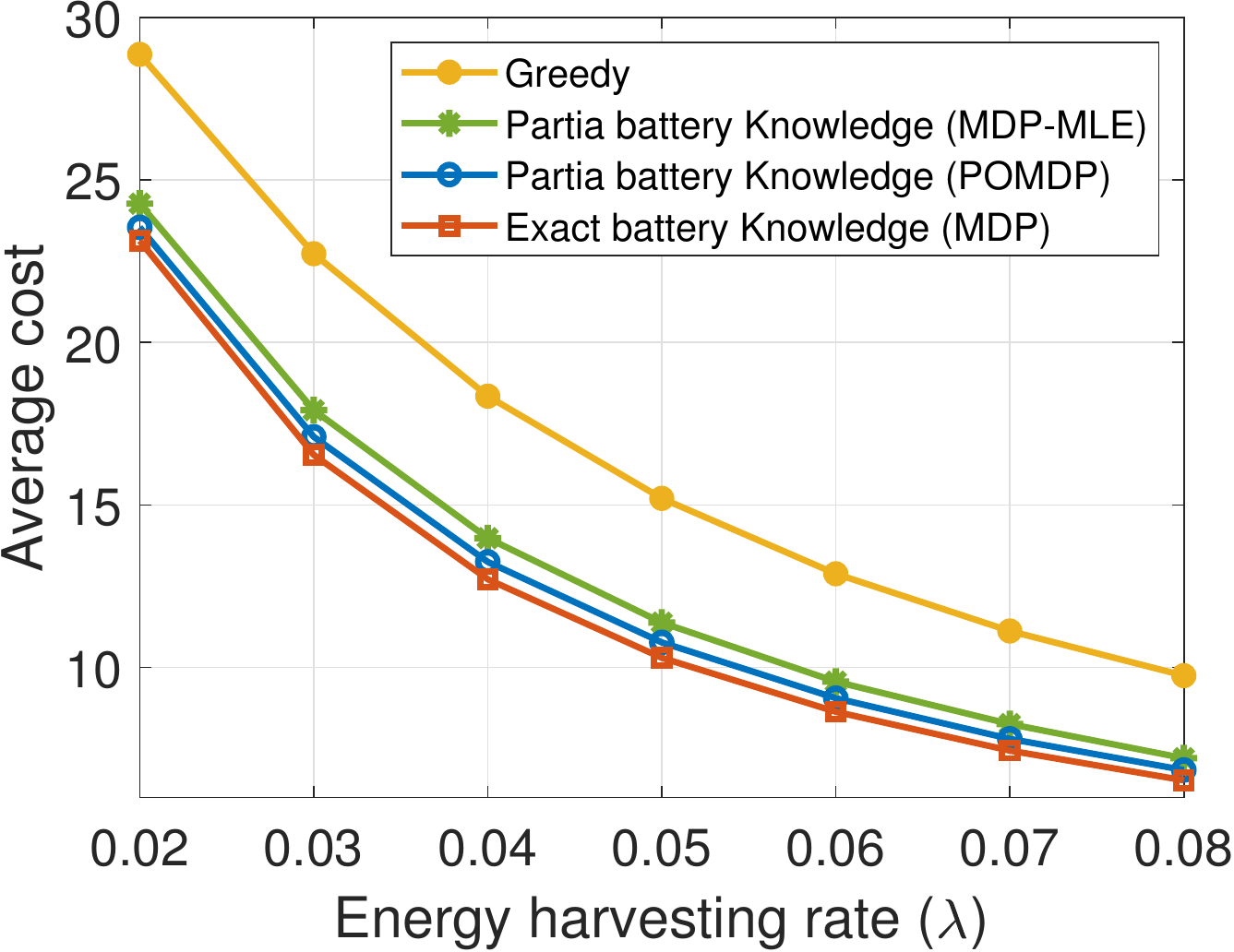}%
		}
		\caption{Average cost with respect to the  energy harvesting rate ($\lambda$) for single sensor scenario.}
		\label{fig_perf_lambdavec}
		\centering
		\subfigure [$\lambda = 0.08$, $B = 2$]{%
			\includegraphics[width=\figPerfwidth \columnwidth]{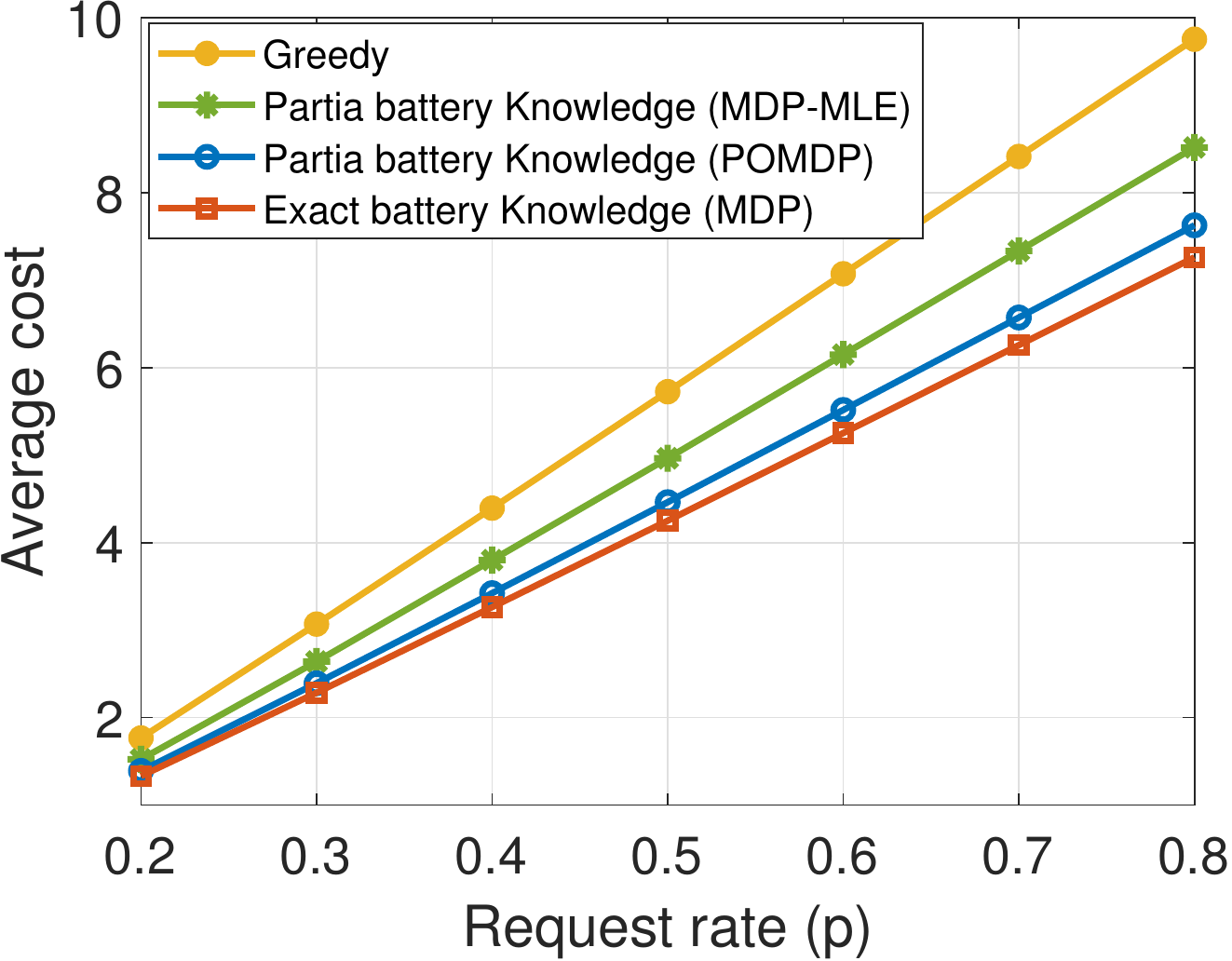}%
		}
		\subfigure [$\lambda = 0.08$, $B = 3$]{%
			\includegraphics[width=\figPerfwidth \columnwidth]{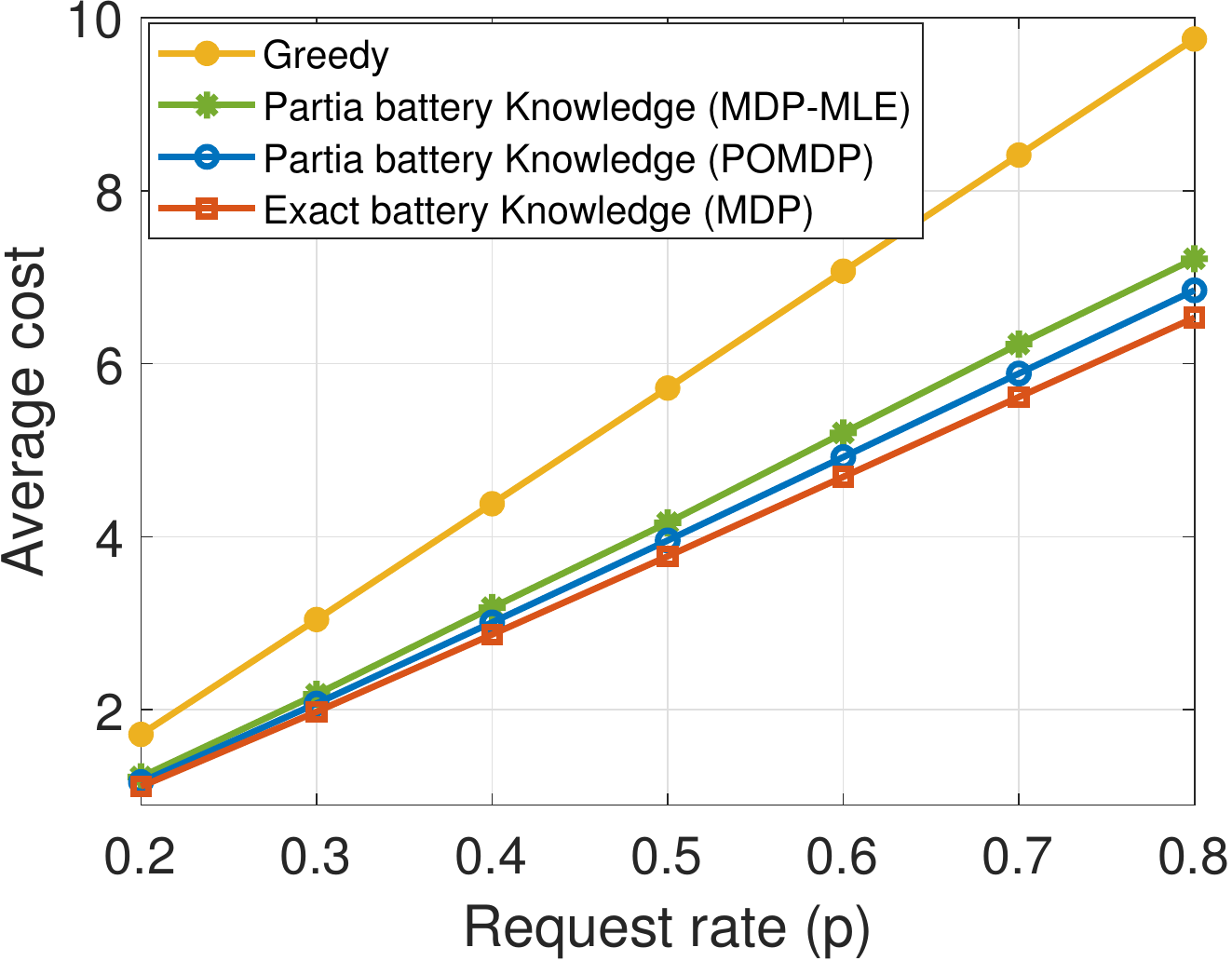}%
		}
		\caption{Average cost with respect to the request rate $p$ for single sensor setup.}
		\label{fig_perf_pvec}
	\end{figure}

	\subsection{Single-sensor IoT Sensing Network}
	We consider a single-sensor scenario with $\lambda = 0.06$, $p = 0.8$, $\Delta^{\mathrm{max}} = 64$, and $B = 2$.
	{Fig.~\ref{fig_structure} illustrates the structure of an optimal policy $\pi^*$, where each point represents a potential  belief-state as a three-tuple $z = (\beta,r,{\Delta})$. For each such $z$, a blue point indicates that the optimal action is to command the sensor (i.e., $\pi^*(z) = 1$), whereas a red point means not to command.}
	{Henceforth, we refer to the set of blue points as the \textit{command region}.}
	We use Table~\ref{tab:belief_set} to represent each belief on the x-axis of these figures; for example, $(0,5)$ is referred to the belief $\boldsymbol{\Lambda}^{5}\beta(1)$ and $(2,3)$ is referred to the belief $\boldsymbol{\Lambda}^{3}{\rho}^{2}$.
	As shown in Fig.~\ref{fig_structure}(a), if there is no request (i.e., $r = 0$), the optimal action is that the edge node does not command the sensor, regardless of the belief and AoI, i.e., $\pi^*(\beta,0,\Delta) = 0$. In this case, the immediate cost (i.e., on-demand AoI \eqref{on-demand-AoI}) becomes zero and the action $a(t) = 0$ leads to energy saving for the sensor, which can be used later to serve the users with fresh measurements.
	Fig.~\ref{fig_structure} illustrates that $\pi^*$ has a \textit{threshold-based} structure with respect to the AoI.
	To exemplify, consider the belief-state $z =((1,7), 1, 22)$
	in which $\pi^*(z) = 1$; then, by the threshold-based structure,  $\pi^*(\underline{z}) = 1$ for all $\underline{z} = ((1,7), 1, \Delta)$, $\Delta \geq 20$.
	From Fig.~\ref{fig_structure}, it can also be inferred that if the optimal action in belief-state $z = (\beta,1,\Delta)$ is $\pi^*(z) = 1$, then the optimal action is $\pi^*(\underline{z}) = 1$ for all states $\underline{z} = ( \boldsymbol{\Lambda}^m\beta, 1, \Delta)$, $m = 1,2,\dots$.
	
	
	Comparing Fig.~\ref{fig_structure}(b) and Fig.~\ref{fig_structure_compare}(a) reveals that the command region enlarges as the EH rate increases. This is because when a sensor harvests energy more frequently, it is able to transmit updates more frequently as well.
	By comparing Fig.~\ref{fig_structure_compare}(b) and Fig.~\ref{fig_structure}(b), it is concluded that the command region shrinks as the request rate $p$ increases, because the edge node commands the sensor less to save its energy for the future requests.

	Fig.~\ref{fig_performance} depicts the performance of the proposed algorithm over time. In the \textit{request-aware greedy} policy, the edge node commands the sensor whenever there is a request (i.e., $r(t) = 1$).
	{As benchmark, we 
		consider a case that the edge node knows the exact battery level at each time slot \cite{hatami2021spawc}. 
		Clearly, this policy
		serves as a lower bound to the proposed POMDP-based algorithm.
		As shown in Fig.~\ref{fig_performance}, for sufficiently large $M$ (e.g., $M \geq 28$ for $\lambda = 0.04$ and $M \geq 16$ for $\lambda = 0.08$), the proposed algorithm obtains optimal performance and reduces the average cost by {approximately
			$25~\%$} compared to the greedy policy.}
	{Besides, the greater the value of $\lambda$, the smaller the value of $M$ for which the proposed algorithm attains optimal performance. This is because the power of the matrix \boldsymbol{$\Lambda$} converges faster.} 
	{Furthermore, the performance of the proposed approach is not too far from the performance under the exact battery knowledge; this relatively small gap shows the impact of the uncertainty about the sensors' battery levels.}
	
	{Fig.~\ref{fig_perf_lambdavec} and Fig.~\ref{fig_perf_pvec} depict the average cost with respect to the energy harvesting rate $\lambda$ and the request rate $p$, respectively.} As expected, the average cost decreases when $\lambda$ increases; the sensor harvests energy more often so that it can send fresh updates more often.
	Further, as the battery capacity $B$ increases, the performance of the MDP-based sub optimal policy (i.e., MLE) becomes closer to the optimal performance, and the gap between the optimal policy for partial battery case and the optimal policy for the exact battery knowledge decreases.
	{Moreover, as shown in Fig.~\ref{fig_perf_pvec}, the average cost increases as $p$ increases. This is because the command region shrinks as the sensor is requested more often.}

	\subsection{Multi-Sensor IoT  Network}\label{sec_simulation_MultiSensor}
	We consider a multi-sensor scenario where ${p_k = 0.8}$, ${\Delta^{\mathrm{max}} = 64}$, and ${B = 3}$.
	{Each sensor is assigned an energy harvesting rate $\lambda_k$ from the set $\{0.01,0.02,\dots, 0.1\}$ sequentially: sensors ${1,11,\ldots}$ have the rate $0.01$, sensors ${2,12,\ldots}$ have the rate $0.02$, and so on.}
	The following benchmarks are used for comparison.
	1) A (request-aware) greedy policy where the edge node commands at most $N$ sensors with the largest AoI from the set $\mathcal{W}(t) \triangleq \{k \mid r_k(t) = 1, k\in\mathcal{K}\}$, i.e., the set of sensors whose status are requested by a user, 2) The lower bound, obtained by following an optimal relaxed policy $\pi^*_{\R}$ (see (12) in \cite{hatami2022POMDP_multisensor}), and 3) the case where the edge node knows the exact battery levels at each slot for which the relax-then truncate approach is used to find an asymptotically optimal policy  \cite{hatami2022JointTcom}.
	
	\newcommand\ffww{.48}
	\begin{figure}[t!]
		\centering
		\subfigure [$\Gamma = 0.02$]{%
			\includegraphics[width=\ffww \columnwidth]{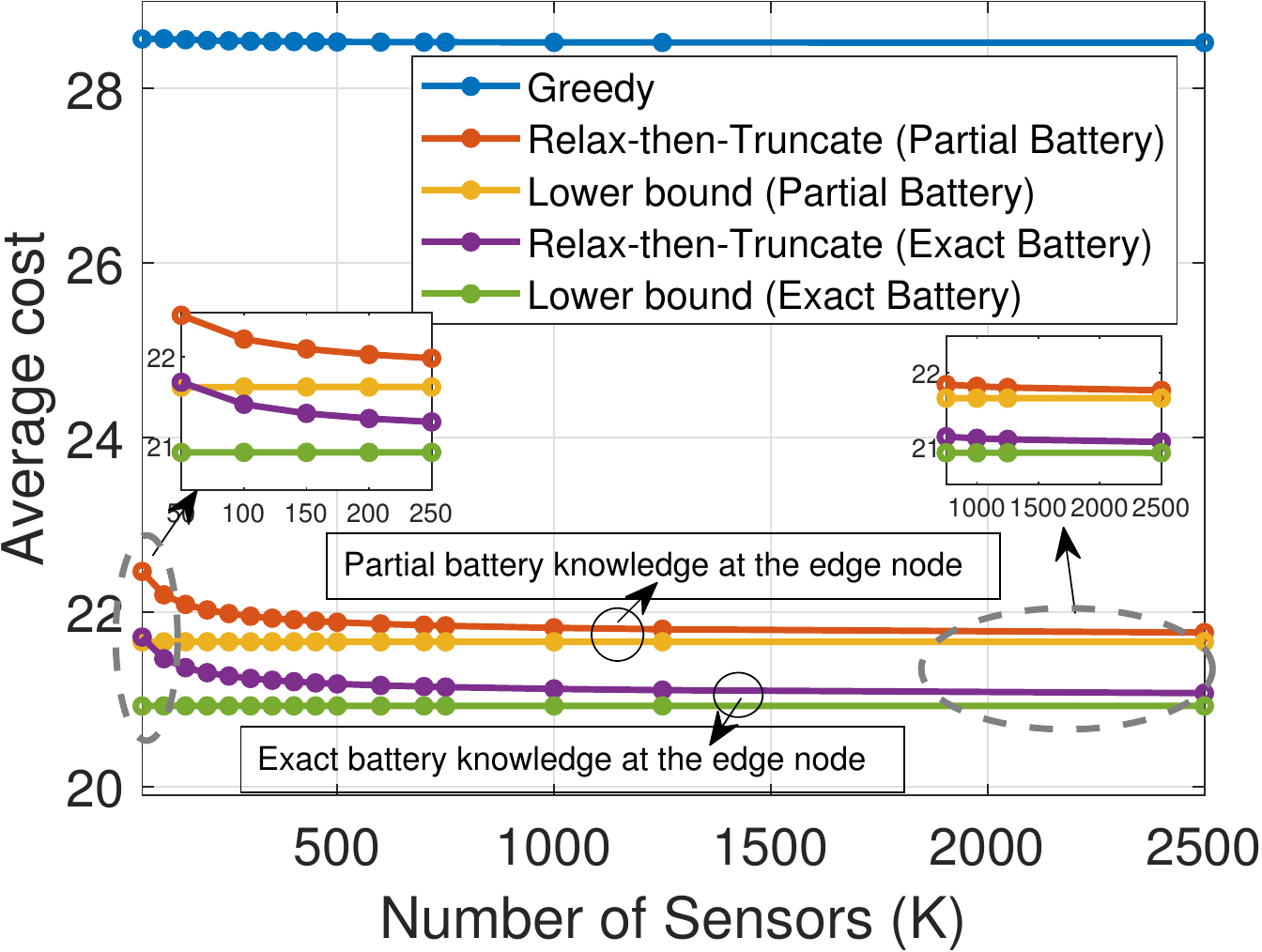}
		}
		\subfigure [$\Gamma = 0.15$]{%
			\includegraphics[width=\ffww \columnwidth]{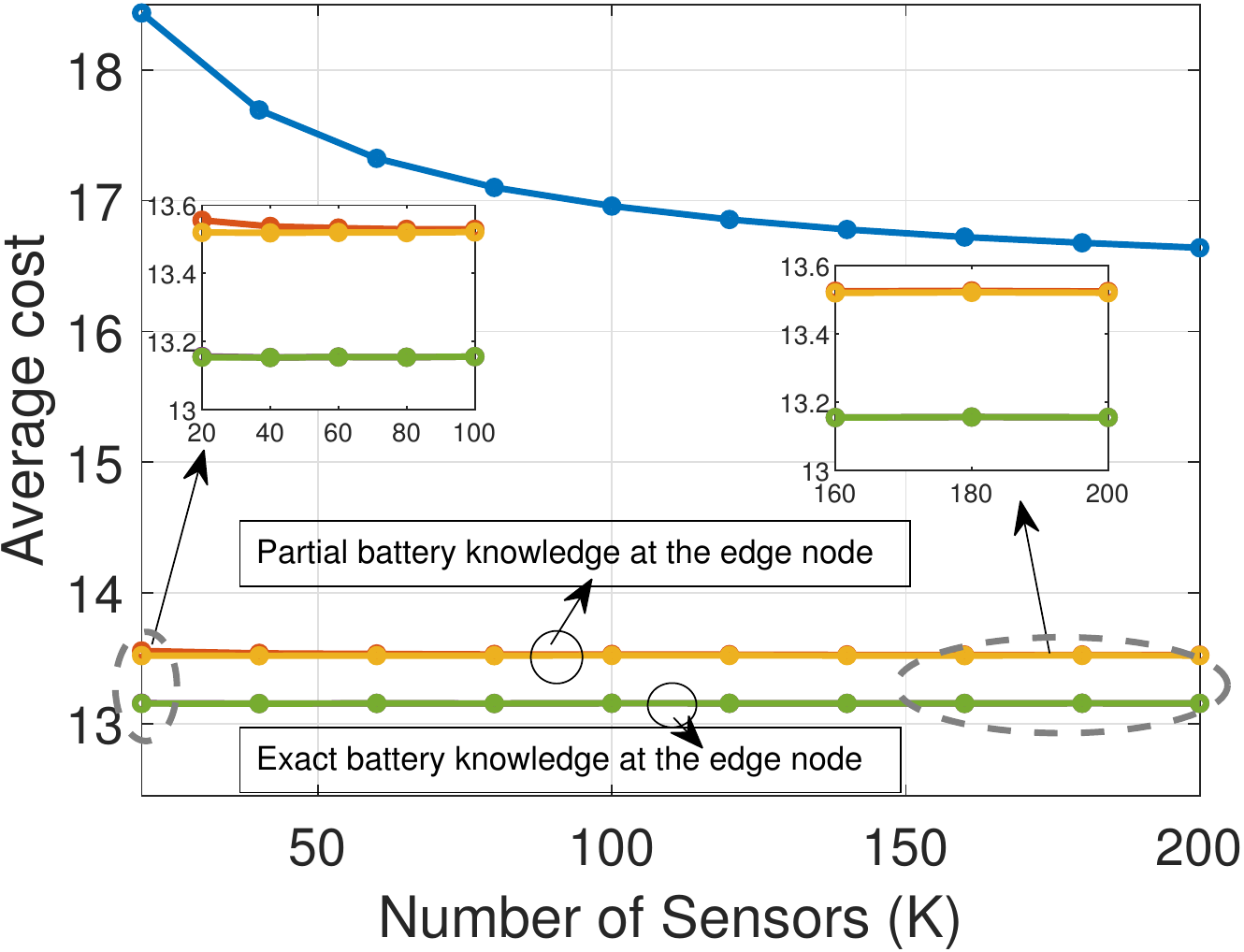}%
		}
		\caption{Performance of the proposed {relax-then-truncate approach} in terms of average cost with respect to the number of sensors $K$ for multi-sensor setup under transmission constraint.
		}
		\label{fig_perf_alpha_fixed}
	\end{figure}

	{The performance of the relax-then-truncate algorithm concerning the number of sensors $K$ for various normalized transmission budget $\Gamma \triangleq \frac{N}{K}$ is shown in Fig.~\ref{fig_perf_alpha_fixed}.}
	The results were acquired by averaging each algorithm over $10$ episodes, each of length $10^7$ slots.
	{First, the proposed algorithm achieves a reduction in the average cost of about $30~\%$ compared to the greedy policy.}
	Due to asymptotic optimality of the proposed algorithm, 
	the gap between the proposed policy and the lower bound is very small for large values of $K$; the same holds true for the exact battery knowledge.
	Interestingly, both relax-then-truncate approaches perform close to the optimal solutions even for moderate numbers of sensors.
	{Moreover, \mbox{$\text{Figs.\ \ref{fig_perf_alpha_fixed}(a) and (b)}$} show that for large $\Gamma$, the proposed policy approaches the optimal performance for smaller values of $K$.}
	{The reason is that, as $\Gamma$ increases, the proportion of sensors that can be commanded at each slot increases, and  thus, the proportion of truncated sensors (i.e., those not commanded under $\tilde{\pi}$ compared to $\pi^*_{\R}$) decreases.}
	
	\begin{figure}[t]
		\centering
		\subfigure[]{
			\includegraphics[width=\ffww \columnwidth]{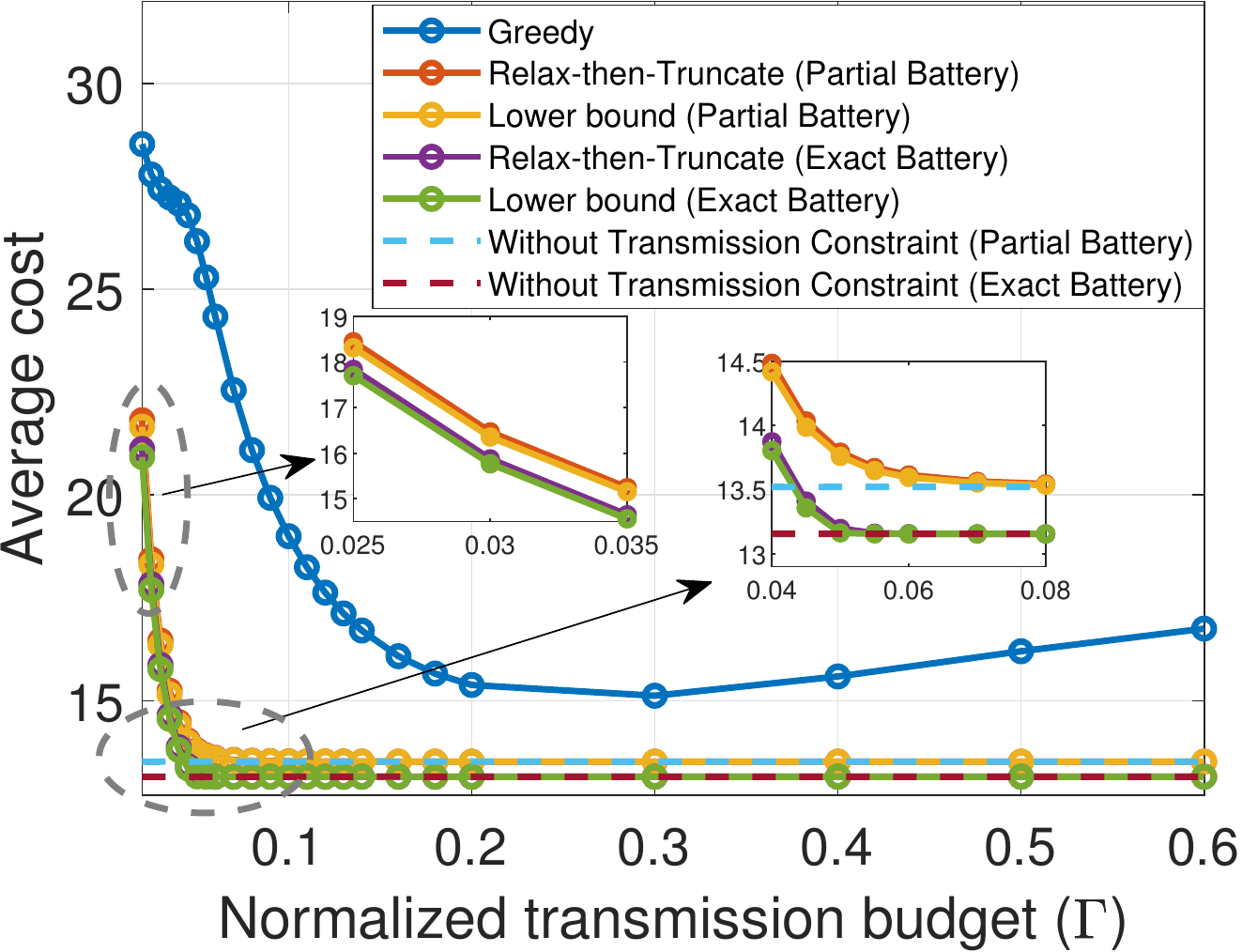}
		}
		\subfigure[]{
			\includegraphics[width=\ffww \columnwidth, trim={0 0 0 0},clip]{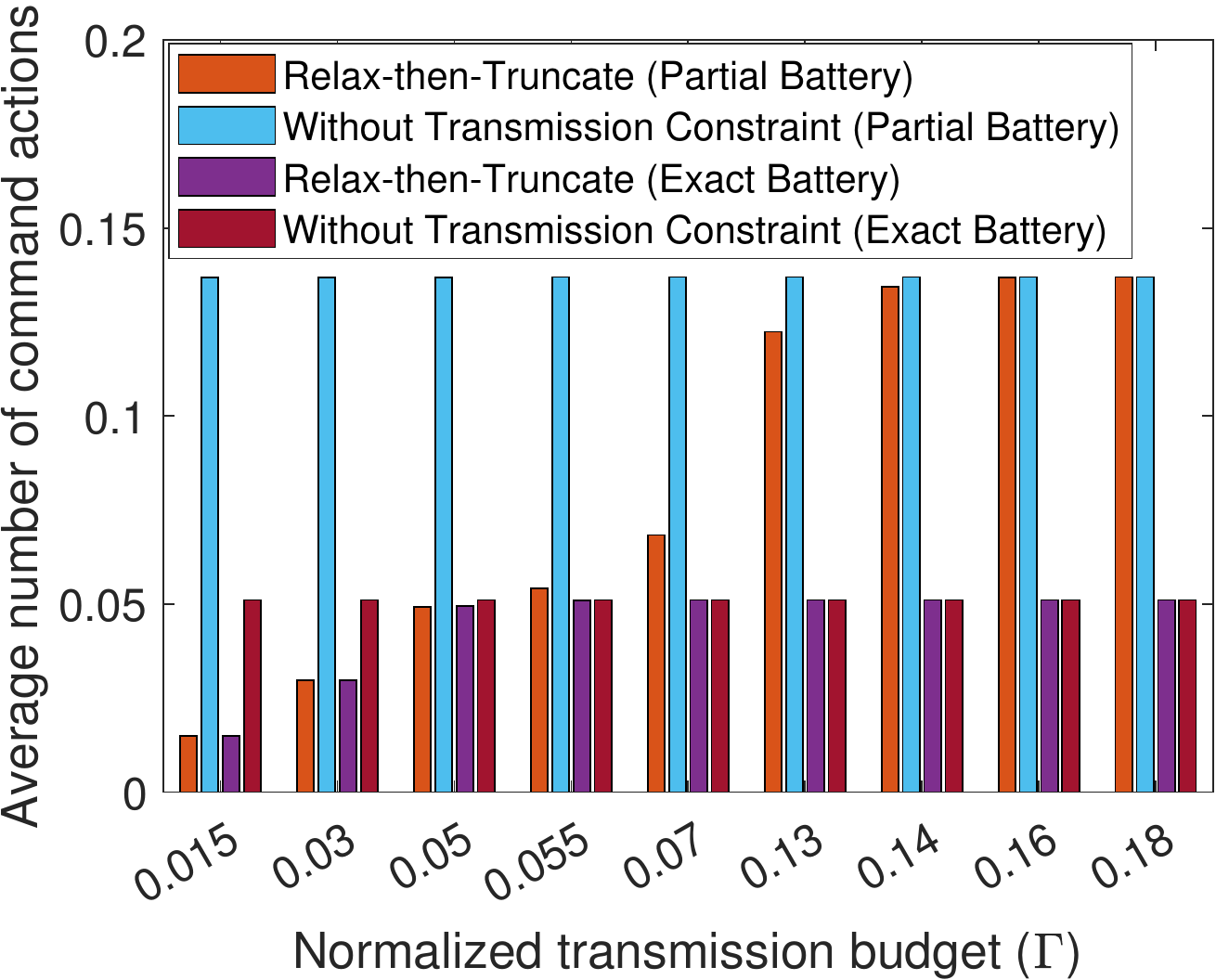}
		}
		\caption{(a) Average cost and (b) Average number of command actions with respect to $\Gamma$ when ${K = 1000}$ for multi-sensor setup under transmission constraint.}
		\label{fig_perf_K_fixed}
	\end{figure}

	
	{The average cost and average number of command actions with respect to the normalized transmission budget $\Gamma$ are illustrated in $\text{Fig.\ \ref{fig_perf_K_fixed}(a)}$ and $\text{Fig.\ \ref{fig_perf_K_fixed}(b)}$, respectively. For the benchmarking, we plot the performance of an optimal policy for the case with no transmission constraint (i.e., ${N \geq K}$) \cite{hatami2021spawc,hatami2022spawc_partialbattery}. As illustrated in $\text{Fig.\ \ref{fig_perf_K_fixed}(a)}$, the average cost for the proposed algorithm decreases as $\Gamma$ increases. This is because, for a fixed $K$, increasing $\Gamma$ increases the transmission budget $N$, allowing the edge node to command more sensors at each slot and serve users with fresh status updates more frequently.}
	{Interestingly, there exists a point beyond which increasing $\Gamma$ does not lead to a decrease in the average cost. This is because, as depicted in $\text{Fig.\ \ref{fig_perf_K_fixed}(b)}$, the average number of command actions ceases to increase (after ${\Gamma \geq 0.055}$ and ${\Gamma \geq 0.16}$ for exact and partial battery knowledge, respectively), meaning that the edge node has more transmission budget than required. In this case, the limited availability of the energy at the EH sensors becomes the primary constraint that limits the transmission of fresh status updates.}
	

	\section{Conclusions}\label{sec_conclusions}
	We investigated status updating under inexact knowledge about the battery levels of the energy harvesting (EH) sensors in an IoT network, where users make on-demand requests to a cache-enabled edge node to send status updates about various random processes monitored by the sensors. Accounting for the partial battery knowledge at the edge node, we derived a POMDP model for the on-demand AoI minimization problem. We converted the POMDP into a belief-state MDP and, via characterizing its key structures, developed an iterative algorithm that obtains an optimal policy for single sensor setup. 
	Additionally, we proposed a sub-optimal MDP-based policy that has less computational complexity than the optimal POMDP-based policy. {We also developed an efficient algorithm implementation leveraging the sparsity of the transition matrices.}
	Furthermore, we extended our approach to the multi-sensor setup under a transmission constraint, where only a limited number of sensors can send status updates at each time slot. In particular, we developed a low-complexity relax-then-truncate algorithm and showed that it is asymptotically optimal as the number of sensors goes to infinity.
	Our numerical experiments showed that an optimal POMDP-based policy has a threshold-based structure,  demonstrated the performance gains obtained by the proposed algorithm compared to a request-aware greedy policy, and depicted that the sub-optimal MDP-based method performs well when the battery capacity of the sensors increases.
	{Furthermore, the performance of the proposed POMDP approach is not too far from the performance under the exact battery knowledge (acts as a lower bound); the relatively small gap shows the impact of the uncertainty about the sensors' battery levels.}
	Finally, our experiments illustrated that the relax-then-truncate method has near-optimal performance even for moderate numbers of sensors in the multi-sensor scenario under the transmission constraint.

	\begin{appendix}
		
		\subsection{Proof of Theorem \ref{lemma_beleief_update}}\label{sec_appendix_lemma_beleief_update}
		We start from the definition of the belief in \eqref{eq_def_update} and express ${\beta_j(t+1)}$ as
		\begin{subequations}\label{eq-proof-belief-update-def}
			\begin{align}
				\beta_j(t+1) & = \Pr\left(b(t+1) = j \mid \phi^\mathrm{c}(t+1)\right) = \Pr\left(b(t+1) = j \mid \phi^\mathrm{c}(t), o(t+1), a(t)\right) 
				\\ & = \frac{\Pr(b(t+1) = j, \phi^\mathrm{c}(t), o(t+1), a(t))}{\Pr(\phi^\mathrm{c}(t), o(t+1), a(t))}  
				\\& = \frac{\Pr(\phi^\mathrm{c}(t),a(t)) \Pr(b(t+1) = j,o(t+1)\mid \phi^\mathrm{c}(t),a(t))}{\Pr(\phi^\mathrm{c}(t),a(t)) \Pr(o(t+1)\mid \phi^\mathrm{c}(t),a(t))} 
				\\& = \frac{\Pr(b(t+1) = j,o(t+1)\mid \phi^\mathrm{c}(t),a(t))}{\Pr(o(t+1)\mid \phi^\mathrm{c}(t),a(t))}
				\\& \overset{(a)}{=} \frac{\Pr(b(t+1) = j,o(t+1)\mid \phi^\mathrm{c}(t),a(t))}{\zeta} 
				\\&= \frac{1}{\zeta}\sum_{i = 0}^{B}\Pr(b(t) = i,b(t+1) = j,o(t+1)\mid \phi^\mathrm{c}(t),a(t))
				\\& = \frac{1}{\zeta}\sum_{i = 0}^{B} \Bigg [\Pr(b(t) = i \mid \phi^\mathrm{c}(t),a(t))\Pr(b(t+1) = j \mid b(t) = i, \phi^\mathrm{c}(t),a(t))\notag\\
				& \hspace{2.5cm}\Pr(o(t+1)\mid b(t+1) = j, b(t) = i, \phi^\mathrm{c}(t),a(t))\Bigg]\\
				& \overset{(b)}{=} \frac{1}{\zeta}\sum_{i = 0}^{B} \Bigg[ \beta_i(t) \Pr(b(t+1)= j \mid b(t) = i,a(t))\notag\\
				& \hspace{2.5cm} \Pr(o(t+1)\mid b(t+1) = j, b(t) = i, \phi^\mathrm{c}(t),a(t))\Bigg]
			\end{align}
		\end{subequations}
		where $(a)$ follows by introducing a normalization factor ${\zeta \triangleq \Pr(o(t+1)\mid \phi^\mathrm{c}(t),a(t))}$, which is calculated using 
		${\sum_j \beta_j(t+1) = 1}$, and $(b)$ follows from i) $\Pr(b(t) = i \mid \phi^\mathrm{c}(t),a(t))=\Pr(b(t) = i \mid \phi^\mathrm{c}(t))$ because $b(t)$ is given when performing action $a(t)$, and subsequently using the belief definition $\beta_i(t)$ in \eqref{eq_def_update}, ii) $\Pr(b(t+1) = j \mid b(t) = i,\phi^\mathrm{c}(t),a(t)) = \Pr(b(t+1) = j \mid b(t) = i,a(t))$ because $b(t+1)$ is independent of $\phi^\mathrm{c}(t)$ given $b(t)$ and $a(t)$.
		Next, we derive an expression for $\beta_j(t+1)$ in \eqref{eq-proof-belief-update-def} for the different cases regarding action ${a(t) \in \{0,1\}}$.

		\subsubsection{Action $a(t) = 0$}
		The edge node does not receive 
		{an update} 
		and thus the next observation is either $o(t+1) = (1,\min\{\Delta(t)+1,\Delta^{\mathrm{max}}\},\tilde{b}(t))$ or $o(t+1) = (0,\min\{\Delta(t)+1,\Delta^{\mathrm{max}}\},\tilde{b}(t))$, which happens with probability $p$ and $1-p$, respectively. Recall that $p$ is the probability of having a request at each slot (i.e., $\mathrm{Pr}\{r(t) = 1\} = p$, $\forall t$). We next calculate the belief update function for 
		$a(t) = 0$ and $o(t+1) = (1,\min\{\Delta(t)+1,\Delta^{\mathrm{max}}\},\tilde{b}(t))$. By \eqref{eq-proof-belief-update-def}, we have
		\begin{equation}\label{eq-proof-belief-update-def_a0}
			\begin{aligned}
				\beta_j(t+1) & =  \frac{1}{\zeta}\sum_{i = 0}^{B} \Bigg [\beta_i(t) \Pr(b(t+1) = j \mid b(t) = i,a(t) = 0)\\
				& \hspace{3cm}\underbrace{\Pr(o(t+1) \mid b(t+1) = j, b(t) = i, \phi^\mathrm{c}(t),a(t)=0)}_{\overset{(a)}{=} p} \Bigg] \\
				& = \frac{p}{\zeta}\sum_{i = 0}^{B} \beta_i(t) \underbrace{\Pr(b(t+1) = j \mid b(t) = i,a(t) = 0)}_{(\star)},
			\end{aligned}
		\end{equation}
		where $(a)$ follows from 
		\begin{equation}\notag\label{eq-proof-belief-update-prob-o1-a0}
			\begin{aligned}
				&\Pr\Bigg(o(t+1) =  (1,\min\{\Delta(t)+1,\Delta^{\mathrm{max}}\},\tilde{b}(t)) \mid   b(t+1)= j, b(t) = i, \\ 
				&\hspace{3cm} \underbrace{(\phi^\mathrm{c}(t-1),r(t),\Delta(t),\tilde{b}(t),a(t-1))}_{\phi^\mathrm{c}(t)},a(t)=0\Bigg)=
				\\ & \Pr\left(o(t+1) =  (1,\min\{\Delta(t)+1,\Delta^{\mathrm{max}}\},\tilde{b}(t)) \mid   \Delta(t), \tilde{b}(t),a(t)=0\right){=}\\
				&{\Pr\left(r(t+1) = 1,\Delta(t+1) = \min\{\Delta(t)+1,\Delta^{\mathrm{max}}\}, \tilde{b}(t+1) = \tilde{b}(t) \mid \Delta(t), \tilde{b}(t),a(t)=0\right)} \overset{(b)}{=} \\
				& \underbrace{\Pr(r(t+1) = 1)}_{=p} \underbrace{\Pr(\Delta(t+1) = \min\{\Delta(t)+1,\Delta^{\mathrm{max}}\},\tilde{b}(t+1) = \tilde{b}(t)  \mid \Delta(t),\tilde{b}(t), a(t)=0)}_{=1} 
				\\&\hspace{.5cm}= p,
			\end{aligned}
		\end{equation}
		where $(b)$ follows from the independence of the request process from the other variables.
		At each slot, the sensor harvests one unit of energy with probability $\lambda$. Thus, $(\star)$ in \eqref{eq-proof-belief-update-def_a0} is expressed as
		\begin{equation} \label{eq-proof-belief-update-Pr-harvest}
			\begin{aligned}
				& \Pr(b(t+1) = j \mid b(t) = i<B,a(t) = 0) = \left\{ 
				\begin{array}{ll}
					1-\lambda, & j = i, \\
					\lambda, & j = i+1,\\
					0, & \mbox{otherwise.}
				\end{array}
				\right.\\
				& \Pr(b(t+1) = j \mid b(t) = B,a(t) = 0) = {\mathds{1}_{\{j = B\}}}.
			\end{aligned}
		\end{equation}
		By substituting \eqref{eq-proof-belief-update-Pr-harvest} into \eqref{eq-proof-belief-update-def_a0}, we can express $\beta_j(t+1)$, for each $j \in \{0,1,\ldots,B\}$, as
		\begin{equation}\label{eq-proof-belief-update-def_a0_summarised}
			\begin{array}{ll}
				&\beta_0(t+1) = \frac{p}{\zeta}(1-\lambda) \beta_0(t),\\
				&\beta_1(t+1) = \frac{p}{\zeta}(\lambda \beta_0(t)+(1-\lambda)\beta_1(t)),\\
				& \dots,\\
				& \beta_{B-1}(t+1) = \frac{p}{\zeta}(\lambda \beta_{B-2}(t)+(1-\lambda)\beta_{B-1}(t)),\\
				& \beta_{B}(t+1) = \frac{p}{\zeta}(\lambda \beta_{B-1}(t)+\beta_{B}(t)).
			\end{array}
		\end{equation}
		Using $\sum_{j = 0}^{B}\beta_j(t+1) = 1$, we can easily calculate the normalization factor to be ${\zeta = p}$.
		By rewriting \eqref{eq-proof-belief-update-def_a0_summarised} in the vector form, the updated belief is given by
		$\beta(t+1) = \boldsymbol{\Lambda} \beta(t)$, where the matrix $\boldsymbol{\Lambda}$ is defined in \eqref{eq_matrix_lambda}. 
		For the case where $a(t) = 0$ and ${o(t+1) = (0,\min\{\Delta(t)+1,\Delta^{\mathrm{max}}\},\tilde{b}(t))}$, one can follow the similar steps and conclude that ${\beta(t+1) = \boldsymbol{\Lambda} \beta(t)}$ as well.

		\subsubsection{Action $a(t) = 1$} For the case where ${a(t) = 1}$, the edge node receives 
		{an update}
		whenever $b(t) \geq 1$ and does not receive an update 
		{if}
		$b(t) = 0$. In this regard, if ${b(t)  = m \geq 1}$, the next observation is either $o(t+1) = (1,1,m)$ or $o(t+1) = (0,1,m)$, ${m \in \{1,2,\dots,B\}}$; and, if $b(t) = 0$, the next observation is either $o(t+1) = (1,\min\{\Delta(t)+1,\Delta^{\mathrm{max}}\},\tilde{b}(t))$ or $o(t+1) = (0,\min\{\Delta(t)+1,\Delta^{\mathrm{max}}\},\tilde{b}(t))$.
		We next calculate the belief update function for these cases. {Starting with} 
		$a(t) = 1$ and ${o(t+1) = (1,\min\{\Delta(t)+1,\Delta^{\mathrm{max}}\},\tilde{b}(t)\})}$, by \eqref{eq-proof-belief-update-def}, we have
		\begin{subequations}\label{eq-proof-belief-update-def_a1_o10}
			\begin{align}
				\beta_j(t+1)  & =  \frac{1}{\zeta}\sum_{i = 0}^{B} \Bigg [\beta_i(t) \Pr(b(t+1) = j \mid b(t) = i,a(t) = 1)\notag
				\\& \hspace{2cm}\Pr(o(t+1)\mid b(t+1) = j, b(t) = i, \phi^\mathrm{c}(t),a(t) = 1)\Bigg] 
				\\& = \frac{1}{\zeta} \beta_0(t) \Pr(b(t+1) = j \mid b(t) = 0,a(t) = 1)\notag
				\\& \hspace{2cm}  \underbrace{\Pr(o(t+1)\mid b(t+1) = j, b(t) = 0, \phi^\mathrm{c}(t),a(t) = 1)}_{\overset{(a)}{=} p } \notag
				\\& \hspace{1cm} + \frac{1}{\zeta} \sum_{i = 1}^{B} \Bigg [\beta_i(t) \Pr(b(t+1) = j \mid b(t) = i,a(t) = 1) \notag
				\\&\hspace{2cm}\underbrace{\Pr(o(t+1)\mid b(t+1) = j, b(t) = i, \phi^\mathrm{c}(t),a(t) = 1)}_{\overset{(b)}{=} 0}\Bigg] \\ 
				&  = \frac{p\beta_0(t)}{\zeta}  \Pr(b(t+1) = j \mid b(t) = 0,a(t) = 1) = \left\{ 
				\begin{array}{ll}
					\frac{p\beta_0(t)}{\zeta} (1-\lambda),&  j = 0,\\
					\frac{p\beta_0(t)}{\zeta} \lambda,&  j = 1,\\
					0,& j = 2,\ldots,B,
				\end{array}
				\right.
			\end{align}
		\end{subequations}
		where $(a)$ follows similarly as \eqref{eq-proof-belief-update-prob-o1-a0} and  $(b)$ follows from
		\begin{equation}\label{eq-proof-belief-update-prob-o11-a1}
			\begin{array}{ll}
				& \Pr\left(o(t+1) =  (1,\min\{\Delta(t)+1,\Delta^{\mathrm{max}}\},\tilde{b}(t)) \mid   b(t+1)= j, b(t) = i \geq 1,  \phi^\mathrm{c}(t),a(t)=1\right)
				\\& = \Bigg[\Pr(r(t+1) = 1) 
				\\& \underbrace{\Pr\!\left(\Delta(t+1) \!=\! \min\{\Delta(t)\!+\!1,\Delta^{\mathrm{max}}\},\tilde{b}(t+1) \!= \!\tilde{b}(t)\! \mid\! b(t+1)\!=\! j, b(t)\! =\! i\! \geq\! 1,\! \phi^\mathrm{c}(t) ,a(t)\!=\!1\!\right)}_{\overset{(c)}{=}0}\Bigg] 
				\\& = 0,
			\end{array}\notag
		\end{equation}
		where $(c)$ is because the edge node receives an update if $a(t) = 1$ and $b(t) \geq 1$, and thus, $\Delta(t+1) = 1$ (see \eqref{eq_AoI}).
		Using $\sum_{j = 0}^{B}\beta_j(t+1) = 1$, we have
		${\zeta = p\beta_0(t)}$.
		By rewriting \eqref{eq-proof-belief-update-def_a1_o10} in the vector form, we 
		have
		$\beta(t+1) = {\rho}^0$ (see \eqref{eq_vectors_rho}).
		For 
		$a(t) = 1$ and $o(t+1) = (0,\min\{\Delta(t)+1,\Delta^{\mathrm{max}}\},\tilde{b}(t))$, one can follow the similar steps and conclude that $\beta(t+1) = {\rho}^0$ as well.

		For the cases where $a(t) = 1$ and ${o(t+1) = (1,1,m)}$, $m \in  \{1,2,\dots,B\}$, by \eqref{eq-proof-belief-update-def}, we have
		\begin{equation}\raisetag{-4ex}\label{eq-proof-belief-update-def_a1_m_summarised}
			\begin{aligned}
				\beta_j(t+1)  & =  \frac{1}{\zeta}\sum_{i = 0}^{B}\Bigg[\beta_i(t) \Pr(b(t+1) = j \mid b(t) = i,a(t) = 1) 
				\\& \hspace{3.5cm}\underbrace{\Pr(o(t+1)\mid b(t+1) = j, b(t) = i, \phi^\mathrm{c}(t),a(t) = 1)}_{\overset{(a)}{=} p \mathds{1}_{\{i=m\}}}\Bigg] 
				\\ & = \frac{p\beta_m(t)}{\zeta}  \Pr(b(t+1) = j \mid b(t) = m,a(t) = 1) 
				\\&= \left\{ 
				\begin{array}{ll}
					\frac{p\beta_m(t)}{\zeta} (1-\lambda),&  j = m-1,\\
					\frac{p\beta_m(t)}{\zeta} \lambda,&  j = m,\\
					0,& \mbox{otherwise.}
				\end{array}
				\right.
			\end{aligned}
		\end{equation}
		where $\mathds{1}_{\{\cdot\}}$ is the indicator function and $(a)$ follows from
		\begin{equation}\notag
			\begin{aligned}
				& \Pr\big(o(t+1) =  (1,1,m) \mid   b(t+1)= j, b(t) = i, \phi^\mathrm{c}(t) ,a(t)=1\big)  
				\\& = \Bigg[ \underbrace{\Pr(r(t+1) = 1)}_{= p}
				\\ & \underbrace{\Pr\big(\Delta(t+1) = 1,\tilde{b}(t+1) = m \geq 1  \mid b(t+1)= j, b(t) = i, \phi^\mathrm{c}(t) ,a(t)=1\big)}_{\overset{(b)}{=} \mathds{1}_{\{i=m\}}}\Bigg] \\& = p \mathds{1}_{\{i=m\}},
			\end{aligned}
		\end{equation}
		where $(b)$ follows because the edge node receives an update
		whenever $a(t) = 1$ and $b(t) \geq 1$, and thus, $\tilde{b}(t+1) = b(t)$ (see Section~\ref{sec_system_model_partial_battery}).
		Using $\sum_{j = 0}^{B}\beta_j(t+1) = 1$, the normalization factor is derived as $\zeta = p\beta_m(t)$.
		Therefore, by \eqref{eq-proof-belief-update-def_a1_m_summarised}, we have $\beta(t+1) = {\rho}^m$, where the vectors ${\rho}^m$, $m \in \{1,2,\dots,B\}$, are defined in \eqref{eq_vectors_rho}. For the cases where $a(t) = 1$ and ${o(t+1) = (0,1,m)}$, $m \in  \{1,2,\dots,B\}$, one can follow the similar steps and conclude that $\beta(t+1) = {\rho}^m$.

		\subsection{Proof of Theorem \ref{theorem_bellman_eq}}\label{sec_appendix_theorem_bellman_eq}
		
		By rewriting the Bellman equation for the average cost POMDP 
		\cite[Chapter~7]{krishnamurthy2016partially}, we have
		\begin{equation}\notag
			\bar{C}^* + h(z) = \min_{a \in \mathcal{A}}[c(z,a) + \sum_{o^\prime}\Pr(o^\prime \mid z,a)h(z^\prime)],~ z \in \mathcal{Z},
		\end{equation}
		where $c(z,a)$ is the immediate cost obtained by choosing action $a$ in belief-state $z$, ${z = (\beta,o) = (\beta,r,\Delta,\tilde{b})}$ is the current belief state, ${o^\prime = (r^\prime,\Delta^\prime,\tilde{b}^\prime)}$ is the observation given action $a$, and ${z^\prime = (\tau(\beta,o^\prime,a),r^\prime,\Delta^\prime,\tilde{b}^\prime)}$ is the next belief state given action $a$ and observation $o^\prime$.
		By defining an action-value function as $Q(z,a)\triangleq c(z,a) + \sum_{o^\prime}\Pr(o^\prime \mid z,a)h(z^\prime)$, we have
		\begin{subequations}\label{eq-proof-bellman-Qfunc}
			\begin{align}
				Q(z,a)  &=  c(z,a) +  \sum_{o^\prime}\Pr(o^\prime \mid z,a)h(z^\prime) 
				\\& = \sum_{s} \Pr(\underbrace{s}_{(b,s^\mathrm{v})} \mid z) c(s,a) + \sum_{o^\prime}\sum_{s} \underbrace{\Pr(o^\prime,s \mid z,a)}_{\Pr(s\mid z,a)\Pr(o^\prime \mid s,z,a)} h(z^\prime)  
				\\& = \sum_{b} \sum_{s^\mathrm{v}}\underbrace{\Pr(b,s^\mathrm{v} \mid \beta,o)}_{\mathds{1}_{\{s^\mathrm{v} = o\}}\Pr(b\mid\beta,s^\mathrm{v},o)} c(b,s^\mathrm{v},a)\notag \\&\hspace{5mm}+\sum_{r^\prime}\sum_{\Delta^\prime}\sum_{\tilde{b}^\prime}\sum_{b}\sum_{s^\mathrm{v}} \underbrace{\Pr(b,s^\mathrm{v}\mid \beta,o,a)}_{\mathds{1}_{\{s^\mathrm{v} = o\}}\Pr(b\mid\beta,s^\mathrm{v},o,a)}\Pr(r^\prime,\Delta^\prime,\tilde{b}^\prime \mid b,s^\mathrm{v},\beta,o,a) h(z^\prime)
				\\ & = \sum_{b} \beta_b c(b,o,a) + \sum_{r^\prime}\sum_{\Delta^\prime}\sum_{\tilde{b}^\prime}\sum_{b} \beta_b \Pr(r^\prime,\Delta^\prime,\tilde{b}^\prime\mid b,\underbrace{r,\Delta,\tilde{b}}_{o},a) h(z^\prime)
				\\ & = \sum_{b} \beta_b c(b,o,a) + \sum_{b} \beta_b \sum_{r^\prime}\Pr(r^\prime) \sum_{\Delta^\prime} \sum_{\tilde{b}^\prime} \Pr(\Delta^\prime,\tilde{b}^\prime\mid b,\Delta,\tilde{b},a) h(z^\prime),
			\end{align}
		\end{subequations}
		where $s^\mathrm{v} = (r,\Delta,\tilde{b})$ is the visible part of the state, which is equivalent to the observation $o$ (i.e., $o = s^\mathrm{v}$).
		For 
		the case where
		$a = 0$, by Theorem~\ref{lemma_beleief_update}, we have $z^\prime = (\tau(\beta,o^\prime,0),r^\prime,\Delta^\prime,\tilde{b}^\prime)={ (\boldsymbol{\Lambda}\beta,r^\prime,\min\{\Delta+1,\Delta^\mathrm{max}\},\tilde{b}^\prime)}$, and thus, $Q(z = (\beta,r,\Delta,\tilde{b}),a = 0)$ in \eqref{eq-proof-bellman-Qfunc} is expressed as
		\begin{subequations}\label{eq_bellman_action0}
			\begin{align}
				Q(z,0) & = Q((\beta,r,\Delta,\tilde{b}),0)  
				\\&= \sum_{b} \beta_b \underbrace{c(b,r,\Delta,\tilde{b},a = 0)}_{= r \min\{\Delta+1,\Delta^\mathrm{max}\}} 
				+ \sum_{b} \beta_b \sum_{r^\prime}\Pr(r^\prime) \sum_{\Delta^\prime} \sum_{\tilde{b}^\prime} \underbrace{\Pr(\Delta^\prime,\tilde{b}^\prime\mid b,\Delta,\tilde{b},a = 0)}_{=\mathds{1}_{\{\Delta^\prime = \min\{\Delta+1,\Delta^\mathrm{max}\},\tilde{b}^\prime = \tilde{b}\}}} h(z^\prime)\\ 
				&  =r \min\{\Delta+1,\Delta^\mathrm{max}\} \underbrace{\sum_{b} \beta_b}_{=1} \notag
				\\&\hspace{1cm} + \sum_{r^\prime}\underbrace{[r^\prime p+ (1-r^\prime)(1-p)]}_{\Pr(r^\prime)} h(\boldsymbol{\Lambda}\beta,r^\prime,\min\{\Delta+1,\Delta^\mathrm{max}\},\tilde{b}) \underbrace{\sum_{b} \beta_b}_{=1}\\ 
				&  =   r \min\{\Delta+1,\Delta^\mathrm{max}\} + \sum_{r^\prime}[r^\prime p+ (1-r^\prime)(1-p)] h(\boldsymbol{\Lambda}\beta,r^\prime,\min\{\Delta+1,\Delta^\mathrm{max}\},\tilde{b}).
			\end{align}
		\end{subequations}
		For the case where $a = 1$, $Q(z = (\beta,r,\Delta,\tilde{b}),a = 1)$ in \eqref{eq-proof-bellman-Qfunc} is expressed as
		\begin{subequations}\label{eq_bellman_action1}
			\begin{align}
				Q(z,1) & = Q((\beta,r,\Delta,\tilde{b}),1)
				\\ & = \sum_{b = 0}^{B} \beta_b c(b,r,\Delta,\tilde{b},a = 1) + \sum_{b = 0}^{B} \beta_b \sum_{r^\prime}\Pr(r^\prime) \sum_{\Delta^\prime} \sum_{\tilde{b}^\prime} \Pr(\Delta^\prime,\tilde{b}^\prime\mid b,\Delta,\tilde{b},a = 1) h(z^\prime) \\
				& = \beta_0\underbrace{c(b = 0,r,\Delta,\tilde{b},a)}_{=r \min\{\Delta+1,\Delta^\mathrm{max}\}} +  \sum_{b = 1}^{B} \beta_b \underbrace{c(b,r,\Delta,\tilde{b},a)}_{= r \times1} \notag
				\\& \hspace{1cm}+\beta_0 \sum_{r^\prime}\Pr(r^\prime) \sum_{\Delta^\prime} \sum_{\tilde{b}^\prime} \underbrace{\Pr(\Delta^\prime,\tilde{b}^\prime\mid b= 0,\Delta,\tilde{b},a)}_{ = \mathds{1}_{\{\Delta^\prime = \min\{\Delta+1,\Delta^\mathrm{max}\},\tilde{b}^\prime = \tilde{b}\}}} h(z^\prime) \notag
				\\ & \hspace{1cm} + \sum_{b = 1}^{B} \beta_b \sum_{r^\prime}\Pr(r^\prime) \sum_{\Delta^\prime} \sum_{\tilde{b}^\prime} \underbrace{\Pr(\Delta^\prime,\tilde{b}^\prime\mid b,\Delta,\tilde{b},a=1)}_{= \mathds{1}_{\{\Delta^\prime = 1 ,\tilde{b}^\prime = b\}}} h(z^\prime) \\
				& = r \beta_0 \min\{\Delta+1,\Delta^\mathrm{max}\} + r \underbrace{\sum_{b = 1}^{B} \beta_b}_{ = 1-\beta_0}  + \Bigg[\beta_0 \sum_{r\prime} [r^\prime p+ (1-r^\prime)(1-p)]\notag \\&\hspace{2cm}h(\underbrace{\tau(\beta,r^\prime,\min\{\Delta+1,\Delta^\mathrm{max}\},\tilde{b},a = 1)}_{\overset{(a)}{=}{\rho}^0},r^\prime,\min\{\Delta+1,\Delta^\mathrm{max}\},\tilde{b})\Bigg] \notag
				\\ & \hspace{1cm}+ \sum_{b = 1}^{B} \beta_b \sum_{r^\prime}[r^\prime p+ (1-r^\prime)(1-p)] h(\underbrace{\tau(\beta,r^\prime,1,b,a)}_{\overset{(a)}{=}{\rho}^b},r^\prime,1,b)
				\\& = r \beta_0 \min\{\Delta+1,\Delta^\mathrm{max}\} + r (1-\beta_0) \notag
				\\& \hspace{10mm}+ \beta_0  \sum_{r^\prime  = 0}^{1} [r^\prime p+ (1-r^\prime)(1-p)] h({\rho}^0,r^\prime,\min\{\Delta+1,\Delta^\mathrm{max}\},\tilde{b})\notag\\
				& \hspace{10mm} + \sum_{b = 1}^{B} \beta_b \sum_{r^\prime = 0}^{1}[r^\prime p+ (1-r^\prime)(1-p)] h({\rho}^b,r^\prime,1,b), 
			\end{align}
		\end{subequations}
		where $(a)$ follows from Theorem~\ref{lemma_beleief_update}.

		\subsection{Proof of Lemma \ref{lemm_struct_LAMBDA}}\label{sec_appendix_lemm_struct_LAMBDA}
		We prove this lemma by mathematical induction. For $m = 1$, we have $\boldsymbol{\Lambda}$ as shown in \eqref{eq_matrix_lambda},
		and hence, the lemma holds for $m = 1$. Assume that the lemma holds for some $m$.
		We prove that the lemma also holds for $m+1$. We have 
		\begin{small}
			\begin{equation}\label{eq_matrix_lambda_power_m+1} \notag
				\begin{array}{ll}
					& \boldsymbol{\Lambda}^{m+1} = \boldsymbol{\Lambda}^m \boldsymbol{\Lambda} \\
					= &
					\begin{pmatrix}
						(1-\lambda)^m & 0 & \cdots & 0 & 0 \\
						m\lambda(1-\lambda)^{m-1} & (1-\lambda)^m & \cdots & 0 & 0\\
						\frac{m(m-1)}{2!}\lambda^2(1-\lambda)^{m-2} & m\lambda(1-\lambda)^{m-1} & \cdots & 0 & 0\\
						\frac{m(m-1)(m-2)}{3!}\lambda^3(1-\lambda)^{m-3} & \frac{m(m-1)}{2}\lambda^2(1-\lambda)^{m-2} & \cdots & 0 & 0\\
						\vdots  & \vdots & \ddots & \vdots & \vdots  \\
						\lambda^{B-1}(1-\lambda)^{m-B+1}\prod_{\nu = 0}^{B-2} \frac{(m-\nu)}{\nu+1}& \lambda^{B-2}(1-\lambda)^{m-B+2}\prod_{\nu = 0}^{B-3} \frac{(m-\nu)}{\nu+1}&  \cdots   & (1-\lambda)^m& 0\\
						1- \sum_{j^\prime = 1}^{B}\Lambda_{j^\prime,1} & 1- \sum_{j^\prime = 1}^{B}\Lambda_{j^\prime,2} & \cdots & 1-(1-\lambda)^m & 1
					\end{pmatrix} \times \\
					&\begin{pmatrix}
						1-\lambda & 0 & \cdots & 0 & 0 \\
						\lambda & 1-\lambda & \cdots & 0 & 0\\
						\vdots  & \vdots  & \ddots & \vdots & \vdots  \\
						0 & 0 & \cdots & 1-\lambda & 0\\
						0 & 0 & \cdots & \lambda & 1
					\end{pmatrix}  \\
					=&
					\begin{pmatrix}
						(1-\lambda)^{m+1} & 0 & \cdots & 0 & 0 \\
						(m+1)\lambda(1-\lambda)^{m} & (1-\lambda)^{m+1} & \cdots & 0 & 0\\
						\frac{(m+1)m}{2!}\lambda^2(1-\lambda)^{m-1} & (m+1)\lambda(1-\lambda)^{m} & \cdots & 0 & 0\\
						\frac{(m+1)m(m-1)}{3!}\lambda^3(1-\lambda)^{m-2} & \frac{(m+1)m}{2}\lambda^2(1-\lambda)^{m-1} & \cdots & 0 & 0\\
						\vdots  & \vdots & \ddots & \vdots & \vdots  \\
						\lambda^{B-1}(1-\lambda)^{m-B+2}\prod_{\nu = 0}^{B-2} \frac{(m+1-\nu)}{\nu+1}& \lambda^{B-2}(1-\lambda)^{m-B+3}\prod_{\nu = 0}^{B-3} \frac{(m-\nu+1)}{\nu+1}&  \cdots   & (1-\lambda)^{m+1}& 0\\
						1- \sum_{j^\prime = 1}^{B}\Lambda_{j^\prime,1} & 1- \sum_{j^\prime = 1}^{B}\Lambda_{j^\prime,2} & \cdots & 1-(1-\lambda)^{m+1} & 1
					\end{pmatrix}.
				\end{array}
			\end{equation}
		\end{small}

		\subsection{Proof of Theorem \ref{theorem_structure_v_part1}}\label{sec_appendix_theorem_structure_v_part1}
		We consider two belief-states ${z = (\beta,r,\Delta,\tilde{b})}$ and ${\zbar = (\beta,r,\Delta,\bbar)}$, where ${\tilde{b} \neq \bbar}$, and prove that $V(z) = V(\zbar)$. 
		{As the sequence $\{V^{(i)}(z)\}_{{i=1,2,\ldots}}$ converges to $V(z)$ for any initialization, {it suffices} to prove that $V^{(i)}(\zbar) = V^{(i)}(z)$, $\forall{i}$, which is established using mathematical induction.}
		The  initial values are selected arbitrarily, e.g., ${V^{(0)}(z) = 0}$ and ${V^{(0)}(\zbar) = 0}$, and thus, $V^{(i)}(z) = V^{(i)}(\zbar)$ holds for ${i = 0}$.
		{Now, suppose that $V^{(i)}(z) = V^{(i)}(\zbar)$ for some $i$. Our goal is to prove that ${V^{(i+1)}(z) = V^{(i+1)}(\zbar)}$ as well.}
		Since, $V^{(i+1)}(z) = \min_{a} Q^{(i+1)}(z,a)$, $\forall z$, we prove that $Q^{(i+1)}(z,a) = Q^{(i+1)}(\zbar,a)$, for all $a \in \{0,1\}$, which concludes the proof. 
		{We provide the proof for $a = 0$; the proof follows similarly for $a = 1$.}
		For $a = 0$, we have 
		\begin{equation}\label{eq_proof_Q0}
			\begin{aligned}
				& Q^{(i+1)}(z,0) -  Q^{(i+1)}(\zbar,0)  = \sum_{r^\prime = 0}^{1} \Bigg[[r^\prime p + (1-r^\prime)(1-p)] 
				\\ &\hspace{8mm}\underbrace{[V^{(i)}(\boldsymbol{\Lambda} \beta,r^\prime,\min\{\Delta+1,\Delta^{\mathrm{max}}\}, \tilde{b}) - V^{(i)}(\boldsymbol{\Lambda} \beta,r^\prime,\min\{\Delta+1,\Delta^{\mathrm{max}}\}, \bbar)]}_{\overset{(a)}{=} 0 }\Bigg] = 0,
			\end{aligned}
		\end{equation}
		where step $(a)$ follows from the induction assumption.
		
		\subsection{Proof of Proposition \ref{theorem-sparsity-transition-matrix}}\label{sec-apndix-theorem-sparsity-transition-matrix}
		For any belief-state $z = (\beta,r,\Delta)$, action $a$,  and next belief state $z^\prime = (\beta^\prime,r^\prime,\Delta^\prime)$, we have
		\begin{equation}\label{eq-sparse-belief-matrix-entries}
			\begin{aligned} 
				&\Pr(z' = (\beta^\prime,r^\prime,\Delta^\prime)\mid z = (\beta,r,\Delta) , a) = \sum_{i = 0}^{B} \Bigg[\Pr(\beta^\prime,r^\prime,\Delta^\prime, b = i\mid \beta,r,\Delta , a)  = \Pr(r^\prime) 
				\\&\hspace{10mm}\sum_{i = 0}^{B}\underbrace{\Pr( b = i\mid \beta,r,\Delta , a,r^\prime)}_{ = \beta_{i}}\Pr(\Delta^\prime\mid \beta,r,\Delta,b = i , a,r^\prime) \Pr(\beta^\prime\mid \beta,r,\Delta , a, r^\prime,\Delta^\prime, b = i)\Bigg]. 
			\end{aligned}
		\end{equation}
		For the case where $a = 0$, we have 
		\begin{equation}\label{eq-sparse-belief-matrix-entries-a0}
			\begin{aligned}
				\Pr(z'\mid z , a ) &  =  \Pr(r^\prime) \sum_{i = 0}^{B} \beta_{i} \underbrace{\Pr(\Delta^\prime\mid \beta,r,\Delta,b = i , a = 0)}_{\overset{(a)}{=}\mathds{1}_{\{\Delta^\prime = \min\{\Delta+1, \Delta^\mathrm{max}\}\}}} \underbrace{\Pr(\beta^\prime\mid \beta,r,\Delta , a = 0, r^\prime,\Delta^\prime, b = i)}_{\overset{(b)}{=}\mathds{1}_{\{\beta^\prime = \Lambda\beta\}}}\\
				& = \left \{ 
				\begin{array}{ll}
					(1-p), &\beta^\prime = \Lambda\beta, r^\prime = 0, \Delta^\prime = \min\{\Delta+1, \Delta^\mathrm{max}\},\\
					p, &\beta^\prime = \Lambda\beta, r^\prime = 1, \Delta^\prime = \min\{\Delta+1, \Delta^\mathrm{max}\},
				\end{array}
				\right.
			\end{aligned}
		\end{equation}
		where $(a)$ is because the AoI increases by one given $a = 0$, and $(b)$ follows from Theorem~\ref{lemma_beleief_update}. {Therefore, $\mathbf{P}^{0}$ has only two non-zero elements in each row, and consequently, ${\mathrm{nz}(\mathbf{P}^{0}) = 2|\mathcal{Z}| = 4M(B+1)\Delta^{\mathrm{max}}}$.}
		For the case where $a = 1$, we have
		\begin{equation}\label{eq-sparse-belief-matrix-entries-a1}
			\begin{aligned}
				\Pr(z'\mid z , a)  & =  \Pr(r^\prime) \Bigg[\beta_{0} \underbrace{\Pr(\Delta^\prime\mid \beta,r,\Delta,b = 0 , a = 1)}_{\overset{(a)}{=}\mathds{1}_{\{\Delta^\prime = \min\{\Delta+1, \Delta^\mathrm{max}\}\}}} \underbrace{\Pr(\beta^\prime\mid \beta,r,\Delta , a = 0, r^\prime,\Delta^\prime, b = 0)}_{\overset{(c)}{=}\mathds{1}_{\{\beta^\prime = \rho^{0}\}}} \\
				& \hspace{1cm}+ \sum_{i =1}^{B} \beta_{i} \underbrace{\Pr(\Delta^\prime\mid \beta,r,\Delta,b = i , a = 1)}_{\overset{(b)}{=}\mathds{1}_{\{\Delta^\prime = 1\}}} \underbrace{\Pr(\beta^\prime\mid \beta,r,\Delta , a = 0, r^\prime,\Delta^\prime, b = i)}_{\overset{(c)}{=}\mathds{1}_{\{\beta^\prime = \rho^{i}\}}} \Bigg]\\
				&  = \left\{
				\begin{array}{ll}
					(1-p)\beta_{0}; &\beta^\prime = \rho^{0}, r^\prime = 0, \Delta^\prime = \min\{\Delta+1, \Delta^\mathrm{max}\}\\
					p\beta_{0}; &\beta^\prime = \rho^{0}, r^\prime = 1, \Delta^\prime = \min\{\Delta+1, \Delta^\mathrm{max}\}\\
					(1-p)\beta_{0}; &\beta^\prime = \rho^{1}, r^\prime = 0, \Delta^\prime = 1\\
					p\beta_{1}; &\beta^\prime = \rho^{1}, r^\prime = 1, \Delta^\prime = 1\\
					\vdots &\vdots \\
					(1-p)\beta_{B}; &\beta^\prime = \rho^{B}, r^\prime = 0, \Delta^\prime = 1\\
					p\beta_{B}; &\beta^\prime = \rho^{B}, r^\prime = 1, \Delta^\prime = 1
				\end{array}
				\right.
			\end{aligned}
		\end{equation}
		where $(a)$ is because the AoI increases by one given $b = 0$, $(b)$ is because the AoI drops to one given $a = 1$ and $b \geq 1$, and $(c)$ follows from Theorem~\ref{lemma_beleief_update}. {Thus, $\mathbf{P}^{1}$ has only $2(B+1)$ non-zero elements in each row, and consequently, $\mathrm{nz}(\mathbf{P}^{0}) = 2(B+1)|\mathcal{Z}| = 4M(B+1)^2\Delta^{\mathrm{max}}$.}
		
	\end{appendix}
	
	
	
	
	\bibliographystyle{IEEEtran}
	\bibliography{Bib/conf_short,Bib/IEEEabrv,Bib/Bibliography}

\end{document}